\documentclass{LMCS}

\def\doi{8 (3:05) 2012}
\lmcsheading%
{\doi}
{1--25}
{}
{}
{Dec.~21, 2011}
{Aug.~10, 2012}
{}

\usepackage{latexsym}
\usepackage{amsthm}
\usepackage{amsmath}
\usepackage{amssymb}
\usepackage{framed}

\usepackage{hyperref,enumerate}
\usepackage{array}
\usepackage{algorithmic}
\usepackage{algorithm}

\usepackage{color}

\begin{document}

\newtheorem{theorem}{Theorem}
\newtheorem{lemma}{Lemma}
\newtheorem{corollary}{Corollary}

\theoremstyle{definition}
\newtheorem{definition}{Definition}
\newtheorem{example}{Example}

\theoremstyle{remark}
\newtheorem{remark}{Remark}

\newcommand{\PS}[1]{\mathcal{P}(#1)}

\renewcommand{\mid}{\,|\,}
\newcommand{\set}[1]{\{#1\}}
\newcommand{\bset}[1]{\bigl\{#1\bigr\}}
\newcommand{\Bset}[1]{\Bigl\{#1\Bigr\}}
\newcommand{\fset}[1]{\left\{#1\right\}}
\newcommand{\df}{:=}
\renewcommand{\emptyset}{\varnothing}

\newcommand{\calf}[1]{\mathcal{#1}}
\newcommand{\cal}[1]{\mathcal{#1}}
\newcommand{\todo}[1]{\textcolor{red}{[\begin{small}\texttt{\textbf{TODO: #1}}\end{small}]}}

\newcommand{\HRule}{\rule{\linewidth}{0.5mm}}

\newcommand{\Prob}{\mathsf{Pr}}
\newcommand{\minweight}{\mathsf{Mwpm}}

\newcommand{\proves}{\vdash}
\newcommand{\witness}[2]{\mathit{Witness}_{#1}^{#2}}
\newcommand{\seq}[1]{\left\langle    #1 \right\rangle }

\newcommand{\It}[1]{\mathit{#1}}
\newcommand{\appref}[1]{Appendix~\ref{app:#1}}
\newcommand{\algref}[1]{Algorithm~\ref{alg:#1}}
\newcommand{\defref}[1]{Definition~\ref{def:#1}}
\newcommand{\theoref}[1]{Theorem~\ref{theo:#1}}
\newcommand{\propref}[1]{Proposition~\ref{prop:#1}}
\newcommand{\corref}[1]{Corollary~\ref{cor:#1}}
\newcommand{\lemref}[1]{Lemma~\ref{lem:#1}}
\newcommand{\exref}[1]{Example~\ref{ex:#1}}
\newcommand{\reref}[1]{Remark~\ref{re:#1}}
\newcommand{\eref}[1]{\eqref{eq:#1}}
\newcommand{\figref}[1]{Figure~\ref{fig:#1}}

\newcommand{\secref}[1]{Section~\ref{sec:#1}}
\newcommand{\eqnref}[1]{Eq.~(\ref{eq:#1})}

\newcommand{\WPHP}{\mathsf{sWPHP}}
\newcommand{\iWPHP}{\mathsf{iWPHP}}
\newcommand{\bPHP}{\mathsf{bPHP}}
\newcommand{\rWPHP}{\mathsf{rWPHP}}
\newcommand{\PHP}{\mathsf{PHP}}

\newcommand{\VPV}{\textit{VPV}}
\newcommand{\PV}{\textit{PV}}

\newcommand{\FP}{\mathsf{FP}}
\newcommand{\ZPP}{\mathsf{ZPP}}
\newcommand{\PP}{\mathisf{PP}}
\newcommand{\PH}{\mathsf{PH}}
\newcommand{\BPP}{\mathsf{BPP}}
\newcommand{\pBPP}{\mathsf{prBPP}}
\newcommand{\coNP}{\mathsf{coNP}}
\newcommand{\CC}{\mathsf{CC}}
\newcommand{\AC}{\mathsf{AC}}
\newcommand{\DET}{\mathsf{DET}}
\newcommand{\RDET}{\mathsf{RDET}}
\newcommand{\NC}{\mathsf{NC}}
\newcommand{\FNC}{\mathsf{FNC}}
\newcommand{\RNC}{\mathsf{RNC}}
\newcommand{\FRNC}{\mathsf{FRNC}}
\newcommand{\NP}{\mathsf{NP}}
\newcommand{\NL}{\mathsf{NL}}
\renewcommand{\L}{\mathsf{L}}
\newcommand{\RsL}{\mathsf{R}\#\mathsf{L}}
\newcommand{\TC}{\mathsf{TC}}
\newcommand{\SL}{\mathsf{SL}}
\newcommand{\APP}{\mathsf{APP}}
\newcommand{\PSPACE}{\mathsf{PSPACE}}
\newcommand{\IP}{\mathsf{IP}}
\newcommand{\numP}{\mathsf{\#P}}
\newcommand{\MA}{\mathsf{MA}}
\newcommand{\MAM}{\mathsf{MA}}
\newcommand{\AM}{\mathsf{AM}}
\newcommand{\pMA}{\mathsf{prMA}}
\newcommand{\pMAM}{\mathsf{prMAM}}
\newcommand{\pAM}{\mathsf{prAM}}
\newcommand{\pP}{\mathsf{prP}}
\newcommand{\BP}{\mathsf{BP}}
\newcommand{\ZPLP}{\mathsf{ZPLP}}
\newcommand{\ZPL}{\mathsf{ZPL}}
\newcommand{\RP}{\mathsf{RP}}
\newcommand{\RL}{\mathsf{RL}}
\newcommand{\coRP}{\mathsf{coRP}}
\newcommand{\coRL}{\mathsf{coRL}}
\newcommand{\BPL}{\mathsf{BPL}}
\newcommand{\coRNC}{\mathsf{coRNC}}
\newcommand{\RAC}{\mathsf{RAC}}
\renewcommand{\P}{\mathsf{P}}
\newcommand{\Ppoly}{\mathsf{P}/\mathsf{poly}}

\newcommand{\VCC}{\mathit{VCC}}
\newcommand{\VAC}{\mathit{VAC}}
\newcommand{\VNC}{\mathit{VNC}}
\newcommand{\VNL}{\textit{VNL}}
\newcommand{\VL}{\mathit{VL}}
\newcommand{\VTC}{\mathit{VTC}}
\newcommand{\VSL}{\mathit{VSL}}
\newcommand{\VsL}{\mathit{V\#L}}
\newcommand{\LFP}{\calf{L}_{\FP}}
\newcommand{\VC}{\mathit{VC}}
\newcommand{\CH}{\mathit{CH}}

\newcommand{\Log}{\mathit{Log}}
\newcommand{\LLog}{\mathit{LogLog}}
\newcommand{\LMAX}{\mathit{LMAX}}
\newcommand{\LMIN}{\mathit{LMIN}}
\newcommand{\LIND}{\mathit{LIND}}
\newcommand{\PIND}{\mathit{PIND}}
\newcommand{\IND}{\mathit{IND}}

\newcommand{\MFRP}{\mathit{MFRP}}
\newcommand{\Det}{\mathsf{Det}}
\newcommand{\Tran}{\mathsf{Transpose}}
\newcommand{\row}{\mathsf{Row}}
\newcommand{\cost}{\mathsf{cost}}
\newcommand{\poly}{\mathsf{poly}}
\newcommand{\mysmash}{\,\#\,}
\newcommand{\onto}{\twoheadrightarrow}
\newcommand{\into}{\hookrightarrow}
\newcommand{\Nz}{\mathsf{numz}}

\newcommand{\bigp}[1]{\left(\begin{array}{ll}#1\end{array}\right)}
\newenvironment{myclaim}[1]
{
	\begin{flushleft}
	\textbf{Claim {#1}:}
}
{
	\end{flushleft}
}

\newenvironment{myclaimb}
{
	\begin{flushleft}
	\textbf{Claim:}
}
{
	\end{flushleft}
}


\numberwithin{equation}{section}
\title{Formalizing Randomized Matching Algorithms}

\author[D.~T.~M.~L\^e]{Dai Tri Man L\^e}	
\address{Department of Computer Science, University of Toronto}	
\email{\{ledt, sacook\}@cs.toronto.edu} 

\author[S.A. Cook]{Stephen A. Cook}

\begin{abstract}
Using Je\v{r}\'abek's framework for probabilistic reasoning,  we formalize the correctness of 
two fundamental $\RNC^{2}$ algorithms for bipartite perfect  matching within the theory $\VPV$ for polytime  reasoning. The first  algorithm is for testing if a bipartite graph has a perfect matching, and is based on the Schwartz-Zippel Lemma  for polynomial identity testing applied to the
Edmonds polynomial of the graph. The second algorithm, due to 
Mulmuley, Vazirani and Vazirani,  is for finding a perfect matching, where the key ingredient of this 
algorithm is  the Isolating Lemma. 

\end{abstract}


\maketitle


\section{Introduction}
There is a substantial literature on theories such as $\PV$, $S^1_2$,
$\VPV$, $V^1$
which capture polynomial
time reasoning \cite{Coo75,Bus86,Kra95,CN10}.  These theories prove
the existence of polynomial time functions, and in many cases they
can prove properties of these functions that assert the correctness
of the algorithms computing these functions.  But in general these
theories cannot prove the existence of probabilistic polynomial
time relations such at those in $\ZPP,\RP, \BPP$ because defining the
relevant probabilities
involves defining cardinalities of exponentially large sets.
Of course stronger theories, those which can define $\numP$ or
$\PSPACE$ functions can treat these probabilities, but such theories
are too powerful to capture the spirit of feasible reasoning.

Note that we cannot hope to find theories that exactly
capture probabilistic complexity
classes such as $\ZPP,\RP,\BPP$ because these are `semantic' classes
which we suppose are not recursively enumerable (cf. \cite{Tha02}).  Nevertheless
there has been significant progress toward developing tools in
weak theories that might be used to describe some of the algorithms
in these classes.

Paris, Wilkie and Woods \cite{PWW88} and Pudl\'ak \cite{Pud91} observed
that we can simulate \emph{approximate counting} in bounded arithmetic
by applying variants of the weak pigeonhole principle.  It seems
unlikely that any of these variants can be proven in the theories for
polynomial time, but they can be proven in Buss's theory $S_ {2}$
for the polynomial hierarchy.
The first connection between the weak pigeonhole principle and randomized 
algorithms was noticed  by Wilkie  (cf. \cite{Kra95}), who showed that
randomized polytime functions witness $\Sigma_{1}^{b}$-consequences of
$S^{1}_{2}+\WPHP(\PV)$ (i.e., $\Sigma_{1}^ {B}$-consequences of
$V^{1}+\WPHP(\LFP)$ in our two-sorted framework), where  $\WPHP(\PV)$ 
denotes the surjective weak pigeonhole principle for all $\PV$
functions (i.e. polytime functions).

Building on these early results, Je\v{r}\'abek \cite{Jer04} showed that we can ``compare''  the sizes 
of two bounded $\Ppoly$ definable sets within $\VPV$  by constructing a surjective 
mapping from  one set to  another. 
Using this method, Je\v{r}\'abek developed tools for describing
algorithms in  $\ZPP$ and $\RP$. He also showed in   \cite{Jer04,Jer05} that the theory  
$\VPV+\WPHP(\LFP)$ is powerful enough  to formalize proofs  of very sophisticated 
derandomization  results, e.g. the Nisan-Wigderson theorem \cite{NW94} and the 
Impagliazzo-Wigderson theorem \cite{IW97}. (Note that Je\v{r}\'abek actually used the single-sorted theory 
$\PV_{1}+\WPHP(\PV)$, but  these two theories are isomorphic.)

In  \cite{Jer07},  Je\v{r}\'abek developed 
an even more systematic approach by showing that for any bounded  $\Ppoly$ definable set, there 
exists a suitable pair of surjective ``counting functions'' which can approximate 
the cardinality of the set up to  a polynomially small error.  From this and other results he argued
convincingly that $\VPV+\WPHP(\LFP)$ is the ``right'' theory for
reasoning about probabilistic polynomial time algorithms.
However so far no one has used his framework for feasible reasoning about
specific interesting randomized algorithms in classes such as $\RP$ and $\RNC^{2}$.

In the present paper we analyze (in $\VPV$) two such algorithms using
Je\v{r}\'abek's framework.  The first one is the $\RNC^{2}$
algorithm for determining whether a bipartite graph has a perfect
matching, based on the Schwartz-Zippel Lemma \cite{Sch80,Zip89} for polynomial
identity testing applied to the Edmonds polynomial \cite{Edm67}
associated with the graph.  The second algorithm, due to Mulmuley,
Vazirani and Vazirani \cite{MVV87}, is in the function class associated
with $\RNC^2$, and uses the Isolating Lemma to find such a perfect
matching when it exists.  Proving correctness of these algorithms
involves proving that the probability of error is bounded above
by $1/2$.  We formulate this assertion in a way suggested by
Je\v{r}\'abek's framework (see Definition \ref{def:prob}).
This involves defining polynomial time functions from
$\set{0,1}^n$ onto $\set{0,1}\times \Phi(n)$, where $\Phi(n)$ is the set of
random bit strings of length $n$ which cause an error in the
computation.  We then show that $\VPV$ proves that the function
is a surjection.

Our proofs are carried out in the theory $\VPV$ for polynomial
time reasoning, without the surjective weak   pigeonhole principle $\WPHP(\LFP)$. 
Je\v{r}\'abek used the $\WPHP(\LFP)$ principle to prove theorems justifying
the above definition of error probability, but we do not need
it  to apply his definition.

Many proofs concerning determinants are based on  the \emph{Lagrange expansion}  
(also known as the Leibniz formula)
\[\Det(A) = \sum_{\sigma\in S_{n}} \mathsf{sgn}(\sigma) \prod_{i=1}^{n} A(i,\sigma(i))\]
where the sum is over exponentially many terms.
Since our proofs in $\VPV$ can only use polynomial time concepts, we cannot
formalize such proofs, and must use other techniques.   In the same vein, the
standard proof of the Schwartz-Zippel Lemma assumes that a multivariate polynomial
given by an arithmetic circuit can be expanded to a sum of monomials.  But this
sum in general has exponentially many terms so again we cannot directly
formalize this proof in $\VPV$.

\section{Preliminaries}

\subsection{Basic bounded arithmetic} \label{s:basicBA}
The theory $\VPV$ for polynomial time reasoning used here
is a two-sorted theory described by Cook and Nguyen \cite{CN10}.
The two-sorted language has variables $x,y,z,\ldots$ ranging over 
$\mathbb{N}$ and variables $X,Y,Z,\ldots$ ranging over finite subsets
of $\mathbb{N}$, interpreted 
as bit strings. Two sorted vocabulary $\mathcal{L}_{A}^{2}$ includes the usual symbols $0,1,+,
\cdot,=,\le$ for arithmetic over $\mathbb{N}$, the length function $|X|$ for strings, the set 
membership relation $\in$, and string equality $=_{2}$ (subscript 2 is usually omitted). 
We will use the notation $X(t)$ for $t\in X$, and think of $X(t)$ as the $t^{\text{th}}$ bit in the string $X$.

The number terms in the base language $\mathcal{L}_{A}^{2}$  are built
from the constants $0,1$, variables $x,y,z,
\ldots$ and length terms $|X|$ using $+$ and $\cdot$. The only string terms are string variables, but 
when we extend $\mathcal{L}_{A}^{2}$ by adding string-valued functions, other string terms will be 
built as usual. The atomic formulas are $t=u$, $X=Y$, $t\le u$, $t\in X$ for any number terms $u,t$ 
and string variables $X,Y$. Formulas are built from atomic formulas using $\wedge,\vee,\neg$ and 
$\exists x$, $\exists X$, $\forall x$, $\forall X$. Bounded number quantifiers are defined as usual, 
and bounded string quantifier $\exists X\le t, \varphi$ stands for $\exists X (|X|\le t\, \wedge \varphi)$ 
and $\forall X\le t, \varphi$ stands for $\forall X (|X|\le t \rightarrow \varphi)$, where $X$ does not 
appear in term $t$.

$\Sigma_{0}^{B}$ is the class of all $\mathcal{L}_{A}^{2}$-formulas with  no string 
quantifiers and only bounded number quantifiers.  $\Sigma_{1}^{B}$-formulas are those of the form 
$\exists \vec{X}< \vec{t}\,\varphi$, where $\varphi\in \Sigma_{0}^{B}$ and the prefix of the bounded 
quantifiers might be empty. These classes are extended to $\Sigma_{i}^{B}$ (and $\Pi_{i}^{B}$)
for all $i\ge 0$, in the usual way.

Two-sorted complexity classes contain relations $R(\vec{x},\vec{X})$, where $\vec{x}$ are number 
arguments and $\vec{X}$ are string arguments. In defining complexity classes using machines or 
circuits, the number arguments are represented in unary notation and the string arguments are 
represented in binary. The string arguments are the main inputs, and the number arguments are 
auxiliary inputs that can be used to index the bits of strings.

In the two sorted setting, we can define $\AC^{0}$ to be the class of relations $R(\vec{x},\vec{X})$ 
such that some alternating Turing machine accepts $R$ in time $\calf{O}(\log n)$ with a constant 
number of alternations, where $n$ is the sum of all the numbers in $\vec{x}$ and 
the total length of all the string arguments in $\vec{X}$. 
Then from the descriptive complexity characterization of $\AC^{0}$, it can 
be shown that a relation $R(\vec{x},\vec{X})$ is in $\AC^{0}$ iff it is represented by some $\Sigma_
{0}^{B}$-formula $\varphi(\vec{x},\vec{X})$.

Given a class of relations $C$, we associate a class $\mathit{FC}$ of string-valued functions $F
(\vec{x},\vec{X})$ and number functions $f(\vec{x},\vec{X})$ with $C$ as follows. We require 
that these functions to be $p$-bounded, i.e., the length of the outputs of $F$  and $f$ is bounded by 
a polynomial in $x$ and $|X|$. Then we define $\mathit{FC}$ to consist of all $p$-bounded number 
functions whose graphs are in $C$ and all $p$-bounded string functions whose bit graphs are in 
$C$.

We write $\Sigma_{i}^{B}(\mathcal{L})$ to denote the class of  $\Sigma_{i}^{B}$-formulas which 
may have function and predicate symbols from $\mathcal{L}\cup \mathcal{L}_{A}^{2}$. A string 
function is $\Sigma_{0}^{B}(\mathcal{L})$-definable  if it is $p$-bounded and its bit graph is 
represented by a $\Sigma_{0}^{B}(\mathcal{L})$-formula. Similarly, a number function is $\Sigma_
{0}^{B}$-definable from $\mathcal{L}$ if it is $p$-bounded and its graph is represented by a $
\Sigma_{0}^{B}(\mathcal{L})$-formula.

The theory $V^{0}$ for $\AC^{0}$ is the basis to develop theories for small complexity classes 
within $\P$ in \cite{CN10}. The theory $V^{0}$ consists of the vocabulary $\mathcal{L}_{A}^{2}$ 
and axiomatized by the sets of \textit{2-BASIC} axioms as given in \figref{2basic}, which express basic properties of symbols 
in $\mathcal{L}_{A}^{2}$, together with the \emph{comprehension} axiom schema 
\begin{align*}
&\Sigma_{0}^{B}(\mathcal{L}_{A}^{2})\textit{-COMP: } &\exists X\le y\, \forall z<y \bigl(X(z) 
\leftrightarrow \varphi(z)\bigr),
\end{align*}
where $\varphi \in \Sigma_{0}^{B}(\mathcal{L}_{A}^{2})$ and $X$ does not occur free in $\varphi$.

\begin{figure}
\HRule
\begin{center}
\begin{minipage}{.45\textwidth}
\begin{align*}
&\textbf{B1. }x+1 \not= 0						\\
&\textbf{B2. }x+1=y+1 \rightarrow x= y			\\
&\textbf{B3. }x+0=x							\\
&\textbf{B4. }x+(y+1)=(x+y)+1					\\
&\textbf{B5. }x\cdot 0 =0						\\
&\textbf{B6. }x\cdot (y+1)  = (x\cdot y)+x			\\
&\textbf{B7. }(x\le y \wedge y\le x)\rightarrow x = y\\
\end{align*}
\end{minipage}
\begin{minipage}{.45\textwidth}
\begin{align*}
&\textbf{B8. }x\le x + y\\
&\textbf{B9. }0\le x\\
&\textbf{B10. }x\le y \vee y\le x\\
&\textbf{B11. }x\le y \leftrightarrow x < y +1\\
&\textbf{B12. }x\not= 0 \rightarrow \exists y\le x \left(y+1=x\right)\\
&\textbf{L1. }X(y)\rightarrow y<|X|				\\
&\textbf{L2. }y+1 = |X| \rightarrow X(y)\\
\end{align*}
\end{minipage}
$\textbf{SE. }\bigl(|X|=|Y| \wedge \forall i <|X|\left(X(i)=Y(i)\right)\bigr) \rightarrow X = Y$
\end{center}
\HRule
\caption{The \textit{2-BASIC} axioms}
\label{fig:2basic}
\end{figure}

In \cite[Chapter 5]{CN10}, it was shown that $V^{0}$ is finitely axiomatizable and  a $p$-bounded 
function is in $\mathsf{FAC}^{0}$ iff it is provably total  in $V^{0}$. A universally-axiomatized 
conservative extension $\overline{V^{0}}$  of $V^{0}$ was also obtained by introducing function symbols and their defining axioms for all $\mathsf{FAC}^{0}$ functions.

In  \cite[Chapter 9]{CN10}, Cook and Nguyen showed how to associate a theory $\VC$ to each 
complexity class $C\subseteq \P$, where $\VC$ extends $V^{0}$ with an additional axiom 
asserting the existence of a solution to a complete problem for $C$.  General techniques are also 
presented for defining a universally-axiomatized conservative extension $\overline{\VC}$ of $\VC$ 
which has function symbols and defining axioms for every function in $\mathit{FC}$, and $\VC$ 
admits induction on open formulas in this enriched vocabulary. It follows from Herbrand's Theorem 
that the provably-total functions in $\overline{\VC}$ (and  hence in $\VC$) are precisely the 
functions in $\mathit{FC}$.  Using this framework, Cook and Nguyen defined explicitly theories for 
various complexity classes within $\P$. 

Since we need some basic linear algebra in this paper, we are interested in the two-sorted theory $
\VsL$ and its universal conservative extension  $\overline{\VsL}$ from \cite{CF10}.  Recall that $\#
\L$ is usually defined as the class of functions $f$ such that for some nondeterministic logspace 
Turing machine $M$, $f(x)$ is the number of accepting computations of $M$ on input $x$. Since 
counting the number of accepting paths of  nondeterministic logspace is $\AC^{0}$-equivalent to 
matrix powering, $\VsL$ was defined to be the extension of the base theory $V^{0}$ with an 
additional axiom stating the existence of powers $A^{k}$ for every matrix
$A$ over $\mathbb{Z}$.  The closure of $\#\L$ under $\AC^0$-reductions
is called $\DET$.  It turns out that computing the determinant of integer
matrices is complete for $\DET$ under $\AC^{0}$-reductions. In fact
Berkowitz's algorithm can be used to reduce the determinant to matrix
powering.  Moreover, $\overline{\VsL}$ proves that the function $\Det$,
which computes the determinant of 
integer matrices based on Berkowitz's algorithm, is in the language of
$\overline{\VsL}$.   Unfortunately it is an open question whether
the theory $\overline{\VsL}$ also proves 
the \emph{cofactor expansion formula} and other basic properties of
determinants.  However from results in \cite{SC04} it follows that
$\overline{\VsL}$ proves that the usual properties of determinants
follow from the Cayley-Hamilton Theorem (which states that a matrix
satisfies its characteristic polynomial).

In this paper, we are particularly interested in  the theory $\VPV$ for polytime reasoning \cite
[Chapter 8.2]{CN10} since we will use it to formalize all of our
theorems.  The universal theory 
$\VPV$ is based on Cook's single-sorted theory $\PV$ \cite{Coo75}, which was  historically the first 
theory designed to capture polytime reasoning.   A nice property of $\PV$ (and $\VPV$) is that their
universal theorems translate into families of propositional tautologies with polynomial size proofs
in any extended Frege proof system.  

The vocabulary $\LFP$ of $\VPV$ extends that of $
\overline{V^{0}}$ with additional symbols introduced based on Cobham's machine independent 
characterization of $\FP$ \cite{Cob65}. Let $Z^{<y}$ denote the first $y$ bits of $Z$. Formally 
the vocabulary  $\LFP$ of $\VPV$ is the smallest set satisfying
\begin{enumerate}[(1)]
\item $\LFP$ contains the vocabulary of $\overline{V^{0}}$
\item For any two  function $G(\vec{x},\vec{X})$, $H(y,\vec{x},\vec{X},\vec{Z})$  over $\LFP$ and  a $\calf
{L}_{A}^{2}$-term $t = t(y,\vec{x},\vec{X})$, if $F$ is defined by \emph{limited recursion} from $G$, 
$H$ and $t$, i.e.,
\begin{align*}
F(0,\vec{x},\vec{X}) &= G(\vec{x},\vec{X}),\\
F(y+1,\vec{x},\vec{X}) &= H(y,\vec{x},\vec{X},F(y,\vec{x},\vec{X}))^{<t(y,\vec{x},\vec{X})},
\end{align*}
then $F\in \calf{L}_{\FP}$.
\end{enumerate}
We will often abuse the notation by letting $\LFP$ denote the set of function symbols in $\LFP$.

The theory $\VPV$ can then be defined to be the theory over $\LFP$ whose axioms are those of  $
\overline{V^{0}}$ together with defining axioms for every function symbols in $\LFP$. $\VPV$ 
proves the scheme $\Sigma_{0}^{B}(\LFP)$-\textit{COMP} and the following schemes
\begin{align*}
\Sigma_{0}^{B}(\LFP)\textit{-IND: }  &&\bigl(\varphi(0)\wedge \forall x \bigl(\varphi(x)\rightarrow \varphi
(x+1)\bigr)\bigr)\rightarrow \forall x \varphi(x)\\
\Sigma_{0}^{B}(\LFP)\textit{-MIN: }  &&\varphi(y)\rightarrow \exists x \bigl(\varphi(x)\wedge \neg 
\exists z<x \varphi(z)\bigr) \\
\Sigma_{0}^{B}(\LFP)\textit{-MAX:}
&&\varphi(0)\rightarrow \exists x \le y \bigl(\varphi(x)\wedge \neg \exists z\le y \bigl(z>x \wedge 
\varphi(z)\bigr)\bigr)
\end{align*}
where $\varphi$ is any $\Sigma_{0}^{B}(\LFP)$-formula.  It follows from Herbrand's Theorem that 
the provably-total functions in $\VPV$ are precisely the functions in $\LFP$.

Observe that $\VPV$ extends $\overline{\VsL}$ since matrix powering can easily be carried out in 
polytime, and thus all theorems of $\overline{\VsL}$ from \cite{CF10,SC04}
are also theorems of $\VPV$.  From results in \cite{SC04} (see 
page~44 of \cite{Jer05} for a correction) it
follows that $\VPV$ proves the Cayley-Hamilton Theorem, and hence
the cofactor expansion formula and other usual properties of determinants of integer matrices.

In our introduction, we mentioned $V^{1}$, the two sorted version of Buss's $S^{1}_{2}$ 
theory \cite{Bus86}. 
The theory $V^{1}$ is also associated with polytime reasoning in the sense that the provably total functions 
of $V^{1}$ are $\FP$ functions, and $V^{1}$ is $\Sigma_{1}^{B}$-conservative over $\VPV$. However,
there is evidence showing that $V^{1}$ is stronger than $\VPV$. For 
example,  the theory $V^{1}$ proves  the $\Sigma_{1}^{B}$\textit{-IND}, $\Sigma_{1}^{B}$\textit{-MIN} and $
\Sigma_{1}^{B}$\textit{-MAX} schemes while $\VPV$ cannot prove these $\Sigma_{1}^{B}$ 
schemes, assuming the polynomial hierarchy does not collapse
\cite{KPT91}.  In this paper we do not use $V^{1}$ to formalize our
theorems, since the weaker theory $\VPV$ suffices  for all our needs.

\subsection{Notation}
If $\varphi$ is a function of $n$ variables, then let $\varphi(\alpha_{1},\ldots,\alpha_{n-1},\bullet)$ 
denote the function of one variable resulting from $\varphi$ by fixing the first $n-1$ arguments  to 
$\alpha_{1},\ldots,\alpha_{n-1}$.   We write $\psi: \Delta \onto \Phi$ to denote that $\psi$ is a 
surjection from $\Delta$ \emph{onto} $\Phi$.

We use $[X,Y)$ to denote $\set{Z\in \mathbb{Z}\mid X\le Z<Y}$, i.e., the interval of integers between 
$X$ and $Y-1$, where strings code integers using signed binary notation.
We also use the standard notation  $[n]$ to denote the set $\set{1,\ldots,n}$.

Given a square matrix $M$, we write $M[i\mid j]$ to denote the $(i,j)$-\emph{minor} of $M$, i.e., the 
square matrix formed by removing the $i$th row and $j$th column from $M$. 

\sloppy
We write $\vec{x}$ to denote number sequence $\seq{x_{1},\ldots,x_{k}}$ and write $\vec{X}$ to denote 
string sequence $\seq{X_{1},\ldots,X_{k}}$. We write $\vec{x}_{k\times k}$ and  $\vec{X}_{k\times k}
$ to denote that $\vec{x}$ and $\vec{X}$ have $k^{2}$ elements and are treated as two-dimensional 
arrays $\seq{x_{i,j}\mid 1\le i,j\le k}$ and  $\seq{X_{i,j}\mid 1\le i,j\le k}$ respectively, where the 
elements of these two-dimensional arrays are listed by rows.
Note that $\vec{x}_{k \times k}$ and $\vec{X}_{k\times k}$ can be simply
encoded as integer matrices, and thus we will use 
matrix notation freely on them.
\fussy

We write the notation ``($T \proves$)'' in front of the statement of a theorem to indicate that the 
statement is formulated and proved within the theory $T$.

\subsection{The weak pigeonhole principle}
The \emph{surjective weak pigeonhole principle} for a function $F$,
denoted by  $\WPHP(F)$,
states that $F$ cannot map  $[0,nA)$ onto $[0,(n+1)A)$. 
Thus, the surjective weak pigeonhole principle for the class of $\VPV$ functions, denoted by $\WPHP(\LFP)$, is the schema \[\bset{\WPHP(F)\mid F\in \LFP}.\]

Note that this principle is believed to be  weaker than the usual surjective ``strong''  pigeonhole principle stating that we cannot map $[0,A)$  onto $[0,A+1)$.  For example, $\WPHP(\LFP)$ can be proven in the theory $V^{3}$ (the two-sorted version of Buss's theory $S^{3}_{2}$) for $\FP^{\Sigma_{3}^{\P}}$ reasoning (cf. \cite{Tha02}), but it is not known if the usual surjective pigeonhole principle for $\VPV$ functions can be proven within the theory $\bigcup_{i\ge 1} V^{i}$ for the polynomial hierarchy  (the two-sorted version of Buss's theory $S_{2}:= \bigcup_{i\ge 1} S_{2}^{i}$ in \cite{Bus86}).

\subsection{Je\v r\'abek's framework for probabilistic reasoning}
\label{s:jerabek}
In this section, we give a brief and simplified overview of Je\v r\'abek's framework \cite{Jer04,Jer05,Jer07} for 
probabilistic reasoning within $\VPV+\WPHP(\LFP)$. For more complete definitions and results, 
the reader is referred  to Je\v r\'abek's work.

Let $F(R)$ be a $\VPV$ 0-1 valued function (which may have other
arguments).  We think of $F$ as defining a relation on binary numbers
$R$.  Let $\Phi(n) = \set{R<2^{n}\mid F(R)=1}.$  Observe that bounding the
probability $\Prob_{R<2^{n}}\bigl[F(R)=1\bigr]$
from above by the ratio $s/t$ is the same as showing that 
$t\cdot |\Phi(n)| \le s\cdot 2^{n}$. More generally,
many probability inequalities can be restated as inequalities between
cardinalities of sets. 
This is problematic since even for the case of polytime definable 
sets,  it follows from Toda's theorem \cite{Tod91} that we cannot express 
their cardinalities directly using bounded
formulas (assuming that the polynomial hierarchy does not collapse).
Hence we need an alternative method 
to compare the sizes of definable sets without exact counting.

The method proposed by Je\v r\'abek in \cite{Jer04,Jer05,Jer07} is based on the following simple observation: 
if $\Gamma(n)$ and $\Phi(n)$ are definable sets and there is a function
$F$ mapping $\Gamma(n)$ onto $\Phi(n)$, then the cardinality of $\Phi(n)$
is at most the cardinality of $\Gamma(n)$. 
Thus instead of counting the sets $\Gamma(n)$ and $\Phi(n)$ directly,
we can compare the sizes of $\Gamma(n)$ and $\Phi(n)$ by showing the existence of a surjection $F$, which in many cases can be easily carried out within weak theories of bounded arithmetic. In this paper we will restrict our discussion to the case when the sets are bounded polytime definable sets and the surjections are polytime functions, all of which can be defined within
$\VPV$, since this is sufficient for our results.

The remaining challenge is then to formally verify that the definition of
cardinality comparison through the use of surjections is a meaningful
and well-behaved definition. The basic properties of surjections like ``any set can be mapped onto itself'' and ``surjectivity is preserved through function compositions'' roughly correspond to the usual reflexivity and transitivity of cardinality ordering, i.e.,  $|\Phi| \le |\Phi|$ and $|\Phi| \le |\Gamma| \le |\Lambda| \rightarrow |\Phi|\le |\Lambda|$ for all bounded definable sets $\Phi$, $\Gamma$ and $\Lambda$. However more sophisticated properties, e.g., dichotomy $|\Phi| \le |\Gamma| \vee |\Gamma| \le |\Phi|$ or ``uniqueness'' of cardinality, turn out to be much  harder to show. 

As a result, Je\v r\'abek proposed in \cite{Jer07} a  systematic and sophisticated framework to 
justify his definition of size comparison. He observed that estimating the size of a $\Ppoly$ 
definable set $\Phi\subseteq [0,2^{n})$ within an error $2^{n}/\poly(n)$ is the same as estimating  $
\Prob_{X \in [0,2^{n})} [X \in \Phi]$ within an error $1/\poly(n)$, which can be solved by drawing $
\poly(n)$ independent random samples $X \in [0,2^{n})$ and check if $X\in \Phi$. This gives us a 
polytime random sampling algorithm for approximating the size of $\Phi$. Since  a counting argument 
\cite{Jer04} can be formalized within $\VPV + \WPHP(\LFP)$ to show the existence of suitable 
average-case hard functions for constructing Nisan-Wigderson generators, this random sampling algorithm 
can be derandomized  to show the existence of an \emph{approximate 
cardinality}  $S$
of $\Phi$ for any given error $E = 2^{n}/\poly(n)$ in the following sense. 
The theory  $\VPV + \WPHP(\LFP)$ proves the existence of $S,y$ and
a pair of $\Ppoly$ ``counting functions'' $(F,G)$ 
\begin{align*}
F:\;&[0,y)\times \bigl(\Phi\uplus [0,E)\bigr)\onto [0,y\cdot S)\\
G:\;&\bigl[0,y\cdot (S+E)\bigr)\onto [0,y)\times \Phi
\end{align*}
Intuitively the pair $(F,G)$ witnesses
that $S-E \le |\Phi| \le S+E$. This allows him to show many properties, 
expected from cardinality comparison, that are satisfied by his method 
within $\VPV + \WPHP(\LFP)$ (see Lemmas 2.10 and 2.11 in \cite{Jer07}).
It is worth noting that proving the uniqueness of cardinality within some error seems
to be the best we can do within bounded arithmetic, where exact counting is not available.\\

For the present paper, the following definition is all we need to
know about Je\v r\'abek's framework.

\begin{defi}\label{def:prob}
Let $\Phi(n) = \{R< 2^{n}\mid F(R)=1\}$, where $F(R)$ is a \VPV\ function
(which may have other arguments) and let $s,t$ be $\VPV$ terms.
Then \[\Prob_{R<2^{n}}\bigl[R\in \Phi(n)]\precsim s/t\]
means that either $\Phi(n)$ is empty, or there exists a $\VPV$
function $G(n,\bullet)$ mapping the set
$[s]\times 2^{n}$ onto the set $[t]\times \Phi(n)$.
\end{defi}

Since we are not concerned with justifying the above definition,
our theorems can be formalized in $\VPV$ without $\WPHP$.

\section{Edmonds' Theorem}
Let $G$ be a bipartite graph with two disjoint sets of vertices $U=\set{u_1,\ldots,u_n}$ and 
$V=\set{v_1,\ldots,v_n}$. We use a pair $(i,j)$ to encode the edge $\set{u_{i},v_{j}}$ of $G$. 
Thus the edge relation of the graph $G$ can be encoded by a boolean matrix $E_{n \times n}$, where we define $(i,j)\in E$, i.e. $E(i,j)=1$, iff $\set{u_{i},v_{j}}$ is an edge of $G$. 

Each perfect matching in $G$ can be encoded by an $n\times n$ permutation matrix $M$ satisfying 
$M(i,j)\rightarrow E(i,j)$ for all $i,j\in [n]$. Recall  that  a \emph{permutation matrix} is a square 
boolean matrix that has exactly one entry of value $1$ in each row and each column and $0$'s 
elsewhere. 

Let $A_{n\times n}$ be the matrix obtained from $G$ by letting $A_{i,j}$ be an indeterminate 
$X_{i,j}$ for all $(i,j)\in E$, and let $A_{i,j}=0$ for all $(i,j)\not\in E$. The matrix of indeterminates 
$A(\vec{X})$ is called  the  \emph{Edmonds matrix} of $G$, and
$\Det(A(\vec{X}))$ is called the 
\emph{Edmonds polynomial}  of $G$.  In general this polynomial has
exponentially many monomials, so for the purpose of proving its properties
in $\VPV$ we consider $\Det(A(\vec{X}))$ to be a function which
takes as input an integer matrix $\vec{W}_{n\times n}$
and returns an integer $\Det(A(\vec{W}))$.   Thus
$\Det(A(\vec{X})) \equiv 0$ means that this function
is identically zero.

The following theorem draws an important   connection between
determinants and matchings.   The standard proof uses the \emph{Lagrange expansion}
which has exponentially many terms, and hence cannot be formalized in $\VPV$.  
However we will give an alternative proof which can be so formalized.

\begin{thm}[Edmonds' Theorem \cite{Edm67}] \label{theo:Edmonds}
($\VPV\proves$) Let $\Det(A(\vec{X}))$ be the Edmonds polynomial of
the bipartite graph $G$. Then $G$ has a perfect matching iff
$\Det(A(\vec{X})) \not\equiv 0$ (i.e. iff there exists an integer matrix
$\vec{W}$ such that 
$\Det(A(\vec{W})) \not= 0$).
\end{thm}

\begin{proof}
For the direction ($\Rightarrow$) we need the following lemma.

\begin{lem}\label{lem:permatrix}
($\VPV\proves$)
$\Det(M) \in \set{-1,1}$ for any  permutation matrix  $M$.
\end{lem}
\begin{proof}[Proof of \lemref{permatrix}]
We will construct a sequence of matrices
\[N_{n}, N_{n-1}, \ldots, N_{1},\]
where $N_{n}=M$, $N_1 = (1)$, and we construct $N_{i-1}$ from $N_{i}$ by choosing $j_{i}$ satisfying $N(i,j_{i})
=1$ and letting $N_{i-1} = N_{i}[i\mid j_{i}]$.

From the way the matrices $N_{i}$ are constructed, we can
easily show by $\Sigma_{0}^{B}(\LFP)$ induction on $\ell=n,\ldots,1$ that the matrices $N_{\ell}$
are permutation matrices.
Finally, using the cofactor expansion formula, we prove by $\Sigma_{0}^{B}(\LFP)$ induction on
$\ell=1,\ldots,n$  that  $\Det(N_{\ell})\in \set{-1,1}$.
\end{proof}

From the lemma we see that if $M$ is the permutation matrix 
representing a perfect matching of $G$, then $\VPV$ proves
$\Det(A(M)) = \Det(M) \in \{1,-1\}$,
so $\Det(A(\vec{X}))$ is not identically 0.

For the direction ($\Leftarrow$) it suffices to describe a polytime
function $F$ that takes as input an integer matrix
$B_{n\times n} = A(\vec{W}) $, where $A(\vec{X})$ is the Edmonds matrix
of a bipartite graph $G$ and $\vec{W}_{n\times n}$ is 
an integer value assignment, and  reason in $\VPV$ that if
$\Det(B)\not= 0$, then $F$ outputs a perfect matching of $G$.

Assume $\Det(B)\not=0$. Note that finding a perfect matching of $G$ is the same
as extracting a nonzero diagonal, i.e., a sequence of nonzero entries $B(1,\sigma(1)),B(2,\sigma(2)),\ldots,B(n,\sigma(n))$, where $\sigma$ is a permutation of the set $[n]$. 
For this purpose, we construct a sequence of matrices
\[B_{n},B_{n-1},\ldots, B_{1},\]
as follows. We let $B_{n}=B$. For $i=n,\ldots,2$, we let  $B_{i-1} = B_{i}[i\mid j_{i}]$ and 
the index $j_{i}$ is chosen using the following method. Suppose we already know $B_{i}$ satisfying 
$\Det(B_{i})\not=0$.  By the cofactor expansion along the last row of $B_{i}$, 
\[ \Det(B_{i}) = \sum_{j=1}^iB_i(i,j)(-1)^{i+j} \Det(B_{i}[i\mid j]).\]
Thus, since $\Det(B_{i})\neq 0$, at least one of the terms in the sum on the right-hand side is nonzero.
Thus, we can choose the least index $j_{i}$ such that $B_i(i,j_{i})\cdot \Det(B_{i}[i\mid j_{i}])\neq 0$.

To extract the perfect matching, we let $Q$ be an $n\times n$ matrix,
where $Q(i,j) = j$.  Then we construct a sequence of matrices 
\[Q_{n},Q_{n-1},\ldots, Q_{1},\]
where $Q_{n}=Q$ and $Q_{i-1}=Q_{i}[i\mid j_{i}]$, i.e., we delete from $Q_{i}$ exactly the row 
and column we deleted from $B_{i}$. 
We define a permutation $\sigma$ by letting 
$\sigma(i)=Q_{i}(i,j_{i})$. Then $\sigma(i)$ is the column
number in $B$ which corresponds to column $j_i$ in $B_i$, and
the set of edges
\[\bset{(i,\sigma(i))\mid 1\le i \le n}\] 
is our desired perfect matching. 
\end{proof}

\section{Schwartz-Zippel Lemma \label{sec:SZ}}
The Schwartz-Zippel Lemma  \cite{Sch80,Zip89} is one of the most 
fundamental tools in the design of randomized 
algorithms. The lemma provides us a $\coRP$ algorithm for the \emph{polynomial identity testing 
problem} (\textsc{Pit}): given an arithmetic circuit  computing a multivariate polynomial $P(\vec{X})$ 
over a  field $\mathbb{F}$, we want to determine if $P(\vec{X})$ is identically zero. 
The \textsc{Pit} problem is important since many problems, 
e.g., primality testing~\cite{AB03}, perfect matching \cite{MVV87}, and 
software run-time testing \cite{WB97}, can be reduced to \textsc{Pit}. Moreover, many fundamental results in complexity theory like $\IP=\PSPACE$ \cite{Sha92} and the PCP theorem \cite{Aro98,AS98} 
make heavy use of \textsc{Pit}  in their proofs. The Schwartz-Zippel lemma can be stated as follows.

\begin{thm}[Schwartz-Zippel Lemma] Let $P(X_1,\ldots,X_n)$ be a non-zero polynomial of 
degree $D \ge 0$ over a field (or integral domain)
$\mathbb{F}$. Let $S$ be a finite subset of $\mathbb{F}$ and let $\vec{R}$ denote the sequence $\seq{R_1,\ldots, R_n}$. Then 
\[\Prob_{\vec{R}\in S^{n}}\bigl[P(\vec{R})=0\bigr]\leq \frac{D}{|S|}.\]
\end{thm}

Using this lemma, we have the following $\coRP$ algorithm for the
\textsc{Pit} problem when $\mathbb{F} = \mathbb{Z}$.
Given a polynomial  \[P(X_1,\ldots,X_n)\] of degree at most $D$, we choose a sequence 
$\vec{R}\in [0,2D)^{n}$ at random.  If $P$ is given implicitly as a
circuit, the degree of $P$ might be exponential, and thus the value of
$P(\vec{R})$ might require exponentially many bits to encode.  In this
case we use the method of Ibarra and Moran \cite{IM83} and let $Y$ be
the result of evaluating $P(\vec{R})$ using arithmetic modulo a random
integer from the interval $[1,D^k]$ for some fixed $k$.  If $Y=0$, then
we report  that $P\equiv 0$. Otherwise, we report that $P\not\equiv 0$.
(Note that if $P$ has small degree, then we can evaluate $P(\vec{R})$
directly.)

Unfortunately the Schwartz-Zippel Lemma seems hard to prove in 
bounded arithmetic. The main challenge is that the degree of $P$ can be exponentially large. Even in the special case when $P$ is given as the 
symbolic determinant of a matrix of indeterminates and hence the degree of $P$ is small, the 
polynomial $P$ still has up to $n!$ terms. Thus, we will focus on a  much weaker version of the 
Schwartz-Zippel Lemma that involves only Edmonds' polynomials since this will suffice for us 
to  establish the correctness of a $\FRNC^{2}$ algorithm for deciding if a bipartite graph 
has a perfect matching.

\subsection{Edmonds' polynomials for complete bipartite graphs}

In this section we will start with the simpler case when every entry of  an Edmonds matrix is a 
variable, since it clearly demonstrates our techniques.
This case corresponds to the Schwartz-Zippel Lemma for Edmonds'
polynomials of complete bipartite graphs.

Let $A$ be the full $n \times n$ Edmonds' matrix $A$, where
$A_{i,j} = X_{i,j}$ for all $1\le i,j \le n$.   We consider the case
that $S$ is the interval of integers $S = [0,s)$
for $s \in \mathbb{N}$, so $|S|=s$.  Then $\Det(A(\vec{X}))$ is a
nonzero polynomial of degree exactly $n$, and we want to show that 
\begin{align*}
\Prob_{\vec{r}\in S^{n^{2}}}\bigl[\Det(A(\vec{r}))=0\bigr] \precsim \frac{n}{s}. 
\end{align*}
Let  
\[\calf{Z}(n,s):=\bset{\vec{r}\in S^{n^2}\mid \Det(A(\vec{r}))=0},\]
i.e., the set of zeros of the Edmonds 
polynomial $\Det(A(\vec{X}))$. Then by \defref{prob}, it suffices to
exhibit a $\VPV$ function mapping 
$[n]\times S^{n^{2}}$  onto $S\times \calf{Z}(n,s)$.  For this it
suffices to give a $\VPV$ function mapping 
$[n]\times S^{n^{2}-1}$ onto $\calf{Z} (n,s)$.
We will define a $\VPV$ function 
\[F(n,s,\bullet):[n]\times S^{n^2-1} \onto \calf{Z}(n,s),\]
so $F(n,s,\bullet)$ takes as
input a pair $(i,\vec{r})$, where $i\in [n]$ and 
$\vec{r}\in S^{n^2-1}$ is a sequence of $n^{2}-1$ elements. 

Let $B$ be an $n\times n$ matrix with elements from $S$.
For $i\in [n]$ let $B_i$ denote the
leading   principal submatrix of $B$ that consists of the $i\times i$
upper-left part of $B$. In other words,
$B_{n} = B$, and $B_{i-1}:=B_{i}[i\mid i]$ for
$i=n,\ldots,2$.  The following fact follows easily from the least
number principle $\Sigma_{0}^{B}(\LFP)\textit{-MIN}$.

\begin{fact}\label{f:detB}($\VPV\proves$)
If $\Det(B)=0$, then there is $i\in [n]$ such that
$\Det(B_j)=0$ for all $i\le j\le n$, and either $i=1$ or $i>1$ and
$\Det(B_{i-1})\ne 0$.
\end{fact}

We claim that given $\Det(B)=0$ and given $i$ as in the fact,
the element $B(i,i)$
is uniquely determined by the other elements in $B$.
Thus if $i=1$ then $B(i,i)=0$, and if $i>1$ then by the cofactor
expansion of $\Det(B_i)$ along row $i$,
\begin{equation}\label{e:fix}
     0 = \Det(B_i) = B_i(i,i) \cdot \Det(B_{i-1}) + \Det(B'_i)
\end{equation}
where $B'_i$ is obtained from $B_i$ by setting $B'(i,i)=0$.
This equation uniquely determines $B_i(i,i)$ because $\Det(B_{i-1})\ne 0$.

The output of $F(n,s,(i,\vec{r}))$ is defined as follows.
Let $B$ be the $n\times n$ matrix determined by the $n^2-1$ elements
in $\vec{r}$ by inserting the symbol $\ast$ (for unknown)
in the position for $B(i,i)$.  Try to use the method above to determine
the value of $\ast= B_i(i,i)$, assuming that $\ast$ is chosen so that
$\Det(B)=0$.  This method could fail because $\Det(B_{i-1}) =0$.
In this case, or if the solution to the equation (\ref{e:fix})
gives a value for $B_i(i,i)$ which is not in $S$,
output the default ``dummy'' zero sequence $\vec{0}_{n\times n}$. 
Otherwise let $C$ be $B$ with $\ast$ replaced by the obtained value of
$B_i(i,i)$.  If $\Det(C) = 0$ then output $C$, otherwise output the
dummy zero sequence.

\begin{thm}\label{theo:Knn}($\VPV\proves$) 
Let $A(\vec{X})$ be the Edmonds matrix of a complete bipartite graph $K_{n,n}$. Let $S$ denote the 
set $[0,s)$. Then the function $F(n,s,\bullet)$ defined above is a
polytime surjection that maps
$[n]\times S^{n^2-1}$ onto  $\calf{Z}(n,s)$.
\end{thm}
\begin{proof} It is easy to see that $F(n,s,\bullet)$ is polytime
(in fact it belongs to the complexity class $\DET$).
To see that $F$ is surjective, let $C$ be an arbitrary matrix in
$\calf{Z} (n,s)$, so $\Det(C)=0$.  Let $i\in [n]$ be determined
by Fact \ref{f:detB} when $B=C$.  Let $\vec{r}$ be the sequence
of $n^2-1$ elements consisting of the rows of $C$ with $C(i,i)$
deleted.  Then the algorithm for computing $F(n,s,(i,\vec{r}))$
correctly computes the missing element $C(i,i)$ and outputs $C$.
\end{proof}

\subsection{Edmonds' polynomials for general bipartite graphs}
For general bipartite graphs, an entry of an Edmonds matrix $A$ might be
$0$, so we cannot simply use leading principal submatrices in our
construction of the surjection $F$. 
However given a sequence $\vec{W}_{n\times n}$ making $\Det(A(\vec{W}))\not=0$, it follows from 
\theoref{Edmonds} that we can find a perfect matching $M$ in polytime. Thus, the nonzero 
diagonal corresponding to the perfect matching $M$ will play the role of the main diagonal in
our construction. The rest of the proof  will proceed similarly.
Thus, we have the following theorem.

\begin{thm}\label{theo:SZEd}
($\VPV \proves$)
There is a $\VPV$ function $H(n,s,A,\vec{W},\bullet)$ where $A_{n\times n}$ is
the  Edmonds matrix for an arbitrary bipartite graph
and $\vec{W}$ is a sequence of $n^2$ (binary) integers, such that if
$\Det(A(\vec{W}))\ne 0$ then $H(n,s,A,\vec{W},\bullet)$ 
maps $[n]\times S^{n^2-1}$ onto
$\set{\vec{r}\in S^{n^2}\mid  \Det(A(\vec{r}))=0}$, where $S=[0,s)$. 
\end{thm}
In other words, it follows from \defref{prob} that the function $H(n,s,A,\vec{W},\bullet)$ in
 the theorem witnesses that
\[\Prob_{\vec{r}\in S^{n^{2}}}\bigl[\Det(A(\vec{r}))=0\bigr] \precsim \frac{n}{s}.\]

\begin{proof}
Assume $\Det(A(\vec{W}))\ne 0$.   Then the polytime function described in
the proof of \theoref{Edmonds} produces an $n\times n$ permutation
matrix $M$ such that for all $i,j\in [n]$, if $M(i,j)=1$ then 
the element $A(i,j)$ in the Edmonds matrix $A$ is not zero.
We apply the algorithm in the proof of Theorem \ref{theo:Knn}, except
that the sequence of principal submatrices of $B$ used in
Fact~\ref{f:detB} is replaced by the sequence $B_n,B_{n-1},\ldots,B_1$
determined by $M$ as follows.  We let $B_n = B$, and for $i=n,\ldots,2$ we let
$B_{i-1} = B_i[i\mid j_i]$, where the indices $j_i$ are chosen the same way as in 
the proof of \theoref{Edmonds} when constructing the perfect matching $M$.
\end{proof}

We note that the mapping $H(n,s,A,\bullet)$ in this case may not be
in $\DET$ since the construction
of $M$ depends on the sequential polytime algorithm from
\theoref{Edmonds} for extracting a perfect matching.

\subsection{Formalizing the $\RNC^{2}$ algorithm for the bipartite perfect 
matching decision problem} \label{s:formBi}

An instance of the \emph{bipartite perfect matching decision} problem
is a bipartite graph $G$ encoded by a matrix $E_{n\times n}$, and we
are to decide if $G$ has a perfect matching.
Here is an $\RDET$ algorithm for the problem. The algorithm is essentially
due to Lov\'asz \cite{Lov79}.  From $E$, construct 
the  Edmonds matrix $A(\vec{X})$ for $G$ and choose a random sequence 
$\vec{r}_{n\times n}\in [2n]^{n^{2}}$. If $\Det(A(\vec{r}))\ne 0$
then we report that $G$ has a perfect matching. Otherwise, we report
$G$ does not have a perfect matching. 

We claim that $\VPV$ proves correctness of this algorithm.
The correctness assertion states that if $G$ has a perfect matching
then the algorithm reports NO with probability at most 1/2, and
otherwise it certainly reports NO.  \theoref{Edmonds} shows that
$\VPV$ proves the latter.   Conversely, if $G$ has a perfect matching given
by a permutation matrix $M$ then the function $H(n,2n,A,M,\bullet)$ 
of \theoref{SZEd} witnesses that the probability of $\Det(A(\vec{r}))=0$
is at most $1/2$, according to Definition \ref{def:prob}, where
$A$ is the Edmonds matrix for $G$.  Hence $\VPV$ proves the correctness
of this case too.

Since $\RDET\subseteq \FRNC^{2}$, this algorithm (which solves a
decision problem) is also an $\RNC^{2}$ algorithm.

\section{Formalizing the Hungarian algorithm}
The Hungarian algorithm is a combinatorial optimization algorithm which solves the \emph
{maximum-weight bipartite matching} problem in polytime and anticipated the later development 
of the powerful \emph{primal-dual method}. The algorithm was developed  by Kuhn \cite{Kuh55}, 
who gave  the name ``Hungarian method'' since it was based on the earlier work of two 
Hungarian mathematicians: D. K\H{o}nig and J. Egerv\'ary. Munkres later reviewed the algorithm 
and showed   that it is indeed polytime  \cite{Mun57}. 
Although the Hungarian algorithm is interesting by itself, 
we formalize the algorithm since we need it in the $\VPV$ proof of the 
Isolating Lemma for perfect matchings in \secref{iso}.

The Hungarian algorithm finds a maximum-weight matching for any
weighted bipartite graph.  The algorithm and its correctness proof are
simpler if we make the  two  following
changes. First, since edges with negative weights can never be in a maximum-weight matching, and thus can be safely deleted, we can assume that every edge has
\emph{nonnegative} weight. Second, by assigning zero 
weight to  every edge  not present, we only need to consider weighted
\emph{complete} bipartite graphs. 

Let $G=(X \uplus Y, E)$ be  a  complete bipartite graph, where $X=\set{x_{i}\mid 1\le i\le n}$ and 
$Y=\set{y_{i} \mid 1\le i\le n}$, and let $\vec{w}$ be an integer weight 
assignment to the edges of $G$, where $w_{i,j}\ge 0$
is the weight of the edge $\set{x_{i},y_{j}}\in E$.  

A pair of integer sequences  $\vec{u}=\seq{u_{i}}_{i=1}^{n}$ and $\vec{v}=\seq{v_{i}}_{i=1}^{n}$ is 
called a \emph{weight cover} if 
\begin{align}
\forall i,j\in [n],\, w_{i,j}\le u_{i} + v_{j}. \label{eq:cover}
\end{align}
The \emph{cost} of a cover is $\cost(\vec{u},\vec{v}) \df \sum_{i=1}^{n} (u_{i} + v_{i})$.
We also define $w(M)\df \sum_{(i,j)\in M} w_{i,j}$. The Hungarian algorithm is based on the 
following important observation.

\begin{lem} \label{lem:hung}
($\VPV\proves $) For any matching $M$ and weight cover $(\vec{u},\vec{v})$, we 
have 
$w(M) \le \cost(\vec{u},\vec{v}).$
\end{lem}
\begin{proof} Since the edges in a matching  $M$ are disjoint, summing the constraints $w_{i,j}\le 
u_{i} + v_{j}$ over all edges of $M$ yields $w(M)  \le \sum_{(i,j)\in M} (u_{i}+v_{j})$.
Since no edge has negative weight, we have $u_{i}+v_{j}\ge 0$ for all $i,j\in [n]$. Thus, 
\[w(M)  \le \sum_{(i,j)\in M} (u_{i}+v_{j}) \le  \cost(\vec{u},\vec{v})\]
for every matching $M$ and every  weight cover $(\vec{u},\vec{v})$.
\end{proof}

Given a weight cover $(\vec{u},\vec{v})$, the \emph{equality subgraph} $H_{\vec{u},\vec{v}}$ 
is the subgraph of $G$ whose vertices are $X\uplus Y$ and whose edges are precisely 
those $\set{x_{i},y_{j}}\in E$ satisfying $w_{i,j} = u_{i}+v_{j} $.

\begin{thm}\label{theo:hung}
($\VPV\proves $) Let $H =H_{\vec{u},\vec{v}}$
be the equality subgraph, and let $M$ be a maximum cardinality matching
of $H$.  Then the following three statements are equivalent
\begin{enumerate}[\em(1)]
\item $w(M) = \cost(\vec{u},\vec{v})$.
\item $M$ is a maximum-weight matching of $G$ and 
$(\vec{u},\vec{v})$ is a minimum-weight cover of $G$.
\item $M$ is a perfect matching of the equality subgraph $H$.
\end{enumerate}
\emph{(cf. \appref{hung} for the full proof of this theorem.)} 
\end{thm}
Below we give a simplified version of the Hungarian algorithm which
runs in polynomial time when the edge weights are small (i.e. presented
in unary notation).  The correctness of the algorithm easily
follows from \theoref{hung}.

\begin{algo}[The Hungarian algorithm] We start with an arbitrary weight
cover $(\vec{u},\vec{v})$ with small weights:
e.g. let \[u_{i} = \max\set{w_{i,j}\mid 1\le j \le n}\] and $v_{i}=0$ for all $i\in[n]$. If the equality 
subgraph $H_{\vec{u},\vec{v}}$ has a perfect matching $M$, we report $M$ as a maximum-weight 
matching of $G$. Otherwise, change the weight cover $(\vec{u},\vec{v})$ as follows. Since the maximum
matching $M$ is  not a perfect matching of $H$, the Hall's condition fails for $H$. 
Thus it is not hard 
(cf. \corref{hall} from \appref{hall})  
to construct in polytime a subset $S \subseteq X$ 
satisfying $|N(S)|<|S|$, where $N(S)$ denotes the  neighborhood of $S$. Hence we  can calculate 
the quantity
\[\delta = \min\bset{u_{i}+v_{j} - w_{i,j}\mid x_{i}\in S \;\wedge\; y_{j}\not\in N(S) },\]
and decrease $u_{i}$ by $\delta$ for all $x_{i}\in S$ and increase $v_{j}$ by $\delta$ for all $y_{j}\in 
N(S)$ without violating the weight cover property~\eref{cover}.
This strictly decreases the sum $ \sum_{i=1}^{n} (u_{i} + v_{i})$.
Thus this process can only repeat at most as many time as the initial cost
of the cover $(\vec{u},\vec{v})$. Assuming that all edge weights
are small (i.e. presented in unary), the algorithm terminates in
polynomial time. Finally we get an 
equality subgraph $H_{\vec{u},\vec{v}}$ containing a perfect matching $M$,
which by \theoref{hung} is  also a maximum-weight matching of $G$.  \\
\end{algo}

When formalizing the Isolating Lemma for bipartite matchings, we need a $\VPV$ function $\minweight$ that takes as 
inputs an edge relation  $E_{n\times n}$ of a bipartite graph $G$ and a nonnegative weight 
assignment  $\vec{w}$ to the edges in $E$, and  outputs a minimum-weight perfect matching if 
such a matching exists, or outputs $\emptyset$ to indicate that no perfect matching exists. 
Recall that the Hungarian algorithm   returns a maximum-weight matching,  and not  a \emph{minimum}-weight \emph{perfect} matching.  However we can use the Hungarian algorithm to compute $\minweight(n,E,\vec{w})$ as follows.  
\begin{algo}[Finding a minimum-weight perfect matching]\label{alg:mwp} $\;$
\begin{algorithmic}[1]
\STATE Let $c = n\cdot\max\bset{w_{i,j} \mid (i,j)\in E} + 1$
\STATE Construct the sequence $\vec{w}'$ as follows 
\[w'_{i,j} = \begin{cases}
c-w_{i,j} 	& \text{if } (i,j)\in E\\
0 	& \text{otherwise.}
\end{cases}\]
\STATE Run the Hungarian algorithm on the complete bipartite graph $K_{n,n}$ with weight 
assignment $\vec{w}'$ to get a maximum-weight matching $M$.
\IF{$M$ contains an edge that is not in $E$} 
\STATE return the empty matching $\emptyset$
\ELSE
\STATE return $M$
\ENDIF
\end{algorithmic}
\end{algo}

Note that since we assign zero weights  to the edges not present and very large weights to other edges, 
the Hungarian algorithm will always prefer the edges that are present in the original bipartite graph. 
More formally for any perfect matching $M$ and non-perfect matching $N$ we have
\begin{align*}
w'(M) 
&\ge nc - n\cdot \max \bset{w_{i,j} \mid (i,j)\in E}\\
&= (n-1)c + \left(c - n\cdot \max \bset{w_{i,j} \mid (i,j)\in E}\right)\\
&= (n-1)c + 1\\
&> (n-1)c\\
&\ge w'(N) 
\end{align*}
The last inequality follows from the fact that $w'_{i,j} \le c$ for all $(i,j)\in N$. Thus, if the Hungarian 
algorithm returns a matching $M$ with at least one edge not in $E$,  then the original graph 
cannot have a perfect matching. Also from the way the weight assignment $\vec{w}'$ was defined, 
every maximum-weight perfect matching of $K_{n,n}$ with weight assignment $\vec{w}'$ is a 
minimum-weight matching of the original bipartite graph.

It is straightforward to check that the above argument can be
formalized in \VPV, so \VPV\ proves the correctness of 
Algorithm \ref{alg:mwp} for computing the function $\minweight$.

\section{$\FRNC^{2}$ algorithm for finding a bipartite perfect matching}
Below we recall the elegant $\FRNC^{2}$ (or more precisely $\RDET$) algorithm due to Mulmuley, 
Vazirani and Vazirani \cite {MVV87} for finding a bipartite perfect matching. Although the original 
algorithm works for  general undirected graphs, we will only focus on bipartite graphs in this paper.

Let $G$ be a bipartite graph with two disjoint sets of vertices $U=\set{u_1,\ldots,u_n}$ and 
$V=\set{v_1,\ldots,v_n}$. We first consider the \emph{minimum-weight bipartite perfect matching problem}, 
where each edge $ (i,j)\in E$ is assigned
an integer weight $w_{i,j}\ge 0$, and we want
to a find a minimum-weight perfect matching of $G$. 
It  turns out there is a $\DET$ algorithm for this problem under two assumptions: the weights 
must be polynomial in $n$, and the  minimum-weight perfect matching must be \emph{unique}. 
We let $A(\vec{X})$ be an Edmonds matrix of the bipartite graph.  
Replace $X_{i,j}$ with $W_{i,j}=2^{w_{i,j}}$  (this is where we need the weights to be small). 
We then compute $\Det(A(\vec{W}))$ using Berkowitz's $\FNC^2$ algorithm. 
Assume that there exists exactly one (unknown) minimum-weight perfect matching $M$. We will 
show in \theoref{Bin2} that $w(M)$ is exactly the position  of the least significant $1$-bit, i.e., 
the number of  trailing zeros, in the binary expansion of $\Det(A(\vec{W}))$.
Once having $w(M)$, we can test if an edge $(i,j)\in E$ belongs to the unique minimum-weight perfect 
matching $M$ as follows. Let $w'$  be the position  of the least significant $1$-bit of 
$\Det(A[i\mid j](\vec{W}))$. We will show in \theoref{Bin3} that the edge $(i,j)$ is in the perfect matching if and only if
$w'$  is precisely $w(M) - w_{i,j}$. Thus, we can test all edges in parallel. 
Note that up to this point, everything can  be done in
$\DET\subseteq \FNC^2$ since the most 
expensive operation is the $\Det$ function, which is complete for $\DET$.

What we have so far is that, assuming that the minimum-weight perfect matching exists and is unique,  there is a $\DET$ algorithm for finding this minimum-weight perfect matching. 
But how do we guarantee that if a minimum-weight perfect matching exists, then it is unique?  
It turns out that we can
assign every edge $ (u_i,v_j)\in E$ a random weight  $w_{i,j}\in [2m]$,
where $m = |E|$, and use the Isolating Lemma \cite{MVV87} to ensure
that the graph has a \emph{unique} 
minimum-weight perfect matching with probability at least $1/2$.

The $\RDET\subseteq \FRNC^{2}$  algorithm for finding a perfect matching is now complete: 
assign random weights  to the edges, and run the $\DET$ algorithm for the unique minimum-weight 
perfect matching  problem. If a perfect matching exists, with probability at least $1/2$, this algorithm 
returns a perfect  matching.

\subsection{Isolating a perfect matching \label{sec:iso}}
We will recall the Isolating Lemma  \cite{MVV87}, the key ingredient of
 Mulmuley-Vazirani-Vazirani $\FRNC^2$ algorithm for finding a perfect matching. 
Let $X$ be a set with $m$ elements $\set{a_{1},\ldots,a_{m}}$ and let $\calf{F}$ be a 
family of subsets  of  $X$. We assign a weight $w_{i}$ to each element $a_{i}\in X$, and  
define the weight of a set $Y \in \calf{F}$ to be $w(Y)\df \sum_{a_{i}\in Y}w_{i}$. 
Let \emph{minimum-weight} be the minimum of the weights of all the sets in $\calf{F}$.
Note that several sets of $\calf{F}$ might achieve  minimum-weight. 
However, if minimum-weight is achieved by a unique $Y\in \calf{F}$, then we 
say that the weight assignment $\vec{w}=\seq{w_{i}}_{i=1}^{m}$ is \emph{isolating} for $\calf{F}$.  
(Every weight assignment is isolating if $|\cal{F}|\le 1$.) 

\begin{thm}[Isolating Lemma \cite{MVV87}] Let $\calf{F}$ be a
family of subsets of an $n$-element 
set $X=\set{a_{1},\ldots,a_{m}}$. Let $\vec{w}=\seq{w_{i}}_{i=1}^{m}$ be a
random weight assignment to the elements in $X$. Then
\[\Prob_{\vec{w}\in [k]^{m}} [\text{$\vec{w}$ is not isolating for $\calf{F}$}] \le \frac{m}{k}.\]
\end{thm}

To formalize the Isolating Lemma in $\VPV$ it seems natural to present
the family $\calf{F}$ by a polytime algorithm.  This is difficult to
do in general (see Remark \ref{r:gen} below), so
we will formalize a special case which suffices to formalize the
$\FRNC^{2}$ algorithm for finding a bipartite perfect matching.
Thus we are given a bipartite graph $G$, and the
family $\calf{F}$ is the set of perfect matchings  of $G$.  We want
to show that if we assign random weights to the edges, then the 
probability that this weight assignment does not isolate a perfect
matching is small. Note that 
although the family $\calf{F}$ here might be exponentially large,
$\calf{F}$ is  polytime definable, 
since recognizing a perfect matching is easy.

\begin{thm}[Isolating a Perfect Matching]\label{theo:iso}
($\VPV \proves$) Let $\calf{F}$ be the family of perfect 
matchings  of a bipartite graph $G$ with edges $E=\set{e_1,\ldots,e_m}$.
Let $\vec{w}$  be a random weight assignment to the edges in $E$. Then
\[\Prob_{\vec{w} \in [k]^{m}} \bigl[\text{$\vec{w}$ is not isolating 
for $\calf{F}$}\bigr] \precsim m/k.\]
\end{thm}

For brevity, we will call a weight assignment $\vec{w}$ ``bad'' 
if $\vec{w}$ is not isolating for $\calf{F}$. Let 
\[\Phi:=\bset{\vec{w}\in [k]^m\mid \vec{w} \text{ is bad for }\calf{F}}.\]
Then to prove \theoref{iso}, it suffices to construct a $\VPV$ function
mapping $[m]\times [k]^ {m-1}$ onto $\Phi$.
Note that the upper bound $m/k$ is independent of the size $n$ of
the two vertex sets.
The set $\Phi$ is polytime definable since $\vec{w} \in \Phi$ 
iff 
\begin{align*}
\exists i,j \in [n]\left(\begin{array}{c}
\text{$E(i,j)$ and  $M(i,j)$ and  $\neg M'(i,j)$ and $M,M'$}\\
\text{encode two perfect matchings with the same weight}
\end{array}\right),
\end{align*}
where  $M$ denotes the output produced by applying the $\minweight$
function (Algorithm \ref{alg:mwp})  on $G$, and 
$M'$ denotes the output produced by  applying $\minweight$ on 
the graph obtained from $G$ by deleting  the edge $(i,j)$.

\begin{proof}[Proof of \theoref{iso}] By \defref{prob} we may assume
that $\Phi$ is nonempty, so there is an element $\delta \in \Phi$.
(We will use $\delta$ as a ``dummy'' element.)   It suffices for us to
construct explicitly a $\VPV$ function 
$\varphi$ mapping $[m]\times [k]^{m-1}$ onto $ \Phi$.
For each $i\in [k]$ we interpret the set $\set{i}\times [k]^{m-1}$
as the set of all possible weight assignments to the 
$m-1$ edges $E\setminus\set{e_i}$. Our function $\varphi$ will map each
set $\set{i}\times [k]^ {m-1}$ onto the set of those bad weight assignments
$\vec{w}$ such that the graph $G$ 
contains two \emph{distinct} minimum-weight perfect matchings $M$ and
$M'$ with $e_i\in M\setminus M'$.

The function $\varphi$ takes as input a sequence 
\[\seq{i,w_1, \ldots,w_{i-1},w_{i+1},\ldots, w_m}\] from $[m]\times [k]^{m-1}$
and does the following. Use the function $\minweight$ (defined by \algref{mwp}) to find a minimum-weight perfect 
matching $M'$ of  $G$ with the edge $e_i$ deleted. 
Use $\minweight$ to find a minimum-weight perfect matching 
$M_{1}$ of the subgraph $G\setminus\set{u_{j},v_{\ell}}$, where  $u_{j}$
and $v_{\ell}$ are the two endpoints of $e_i$. 
If both perfect matchings $M'$ and $M_{1}$ exist and satisfy
$w(M')-w(M_{1})\in [k]$, then $\varphi$ outputs  the sequence
\begin{equation}\label{e:badOut}
\seq{w_1, \ldots,w_{i-1},w(M')-w(M_{1}),w_{i+1},\ldots, w_m}. 
\end{equation}
Otherwise $\varphi$ outputs the dummy element $\delta$ of $\Phi$.

Note that if both $M'$ and $M_{1}$ exist, then (\ref{e:badOut}) is
a bad weight assignment, since $M'$ and $M = M_{1} \cup \{e_i\}$ are
distinct minimum-weight perfect matchings of $G$ under this assignment.

To show that $\varphi$ is surjective, consider an arbitrary bad weight
assignment $\vec{w}=\seq {w_i}_{i=1}^{m}\in \Phi$. Since $\vec{w}$ is bad,
there are two distinct minimum-weight perfect 
matchings $M$ and $M'$  and some edge $e_i\in M\setminus M'$. Thus
from how $ \varphi$ was defined, 
\[\seq{i,w_1,\ldots,w_{i-1},w_{i+1},\ldots,w_m}\in [m]\times [k]^{m-1}\] 
is an  element that gets mapped to the bad weight assignment $\vec{w}$.
\end{proof}

\begin{rem}\label{r:gen}
The above proof uses the fact that there is a polytime algorithm for
finding a minimum-weight perfect matching (when one exists) in
an edge-weighted bipartite graph.  This suggests limitations on
formalizing a more general version of Theorem \ref{theo:iso} in
$\VPV$.  For example, if $\calf{F}$ is the set of Hamiltonian
cycles in a complete graph, then finding a minimum weight member
of $\calf{F}$ is $\NP$ hard.  
\end{rem}

\subsection{Extracting the unique minimum-weight perfect matching}
Let $G$ be a bipartite graph and assume that $G$ has a perfect matching. Then in \secref{iso} 
we formalized a version of the Isolating Lemma, which with high probability gives 
us a weight assignment $\vec{w}$ for which $G$ has a unique minimum-weight perfect matching. 
This is the first step of the Mulmuley-Vazirani-Vazirani algorithm. Now we
proceed with the second step, where we need to output this minimum-weight perfect matching using 
a $\DET$  function.  

Let $B$ be the matrix we get by substituting  $W_{i,j} = 2^{w_{i,j}}$ for each nonzero  entry $(i,j)$ 
of the Edmonds matrix  $A$ of $G$.  We want to show that if $M$ is the unique minimum  weight perfect 
matching of $G$ with respect  to $\vec{w}$,  then the weight $w(M)$  is exactly the position 
of the least  significant $1$-bit in the binary expansion of $\Det(B)$. 
The usual proof of this fact is not hard, but it uses properties of the Lagrange expansion 
for the determinant,  which has 
exponentially many terms and hence cannot be formalized in $\VPV$.
Our proof avoids using the Lagrange expansion,  and utilizes
properties of the cofactor expansion instead.

\begin{lem} ($\VPV \proves$)\label{lem:Bin1}
There is a $\VPV$ function that takes as inputs an  $n\times n$ Edmonds' matrix $A$ and a weight sequence \[\vec{W}=\seq{W_{i,j} = 2^{w_{i,j}}\mid 1\le i,j\le n}.\] And if
$B=A(\vec{W})$ satisfies $\Det(B)\not=0$ and $p$ is the position of the least significant $1$-bit of $\Det(B)$,  then the $\VPV$ function outputs a perfect matching $M$ of weight 
at most $p$.
\end{lem}
It is worth noting that the lemma holds regardless of whether or not the bipartite graph corresponding to $A$ and weight assignment $\vec{W}$ has a unique minimum-weight perfect matching.

The proof of \lemref{Bin1} is very similar  to that of \theoref{Edmonds}. Recall that in  \theoref{Edmonds},
given a matrix $B$ satisfying $\Det(B)\not=0$, we  want to extract a nonzero  diagonal of $B$. 
In this lemma, we are  given the position $p$ of the least  significant $1$-bit of $\Det(B)$, 
and we want to get a nonzero diagonal of  $B$ whose product has the 
least significant $1$-bit at position at most $p$. For this, we can use the  same method of extracting 
the nonzero diagonal from \theoref{Edmonds} with the following modification. When 
choosing a term of the Lagrange expansion on the recursive step, we will also need to make sure 
the chosen term produces a nonzero sub-diagonal of $B$ that will not contribute too much weight 
to  the diagonal  we are trying to extract. This ensures that the least significant $1$-bit of the weight of 
the chosen diagonal is  at   most $p$. 


For the rest of this section, we define $\Nz(Y)$ to be the position of the least significant 1-bit of the binary string $Y$.
Thus if $\Nz(Y) = q$ then $Y = \pm 2^qZ$ for some positive odd integer $Z$.

\begin{proof}[Proof of \lemref{Bin1}]
We construct a sequence of matrices 
\[B_{n},B_{n-1},\ldots, B_{1}\]
where $B_{n}=B$ and $B_{i-1} = B_{i}[i \mid j_{i}]$ for $i=n\ldots, 2$
where the index $j_{i}$ is chosen as follows. Define
\[ T_{i} := \left(\prod_{\ell = i+1}^{n} B_{\ell}(\ell,j_{\ell})\right)\Det(B_{i}).\]
Assume we are given $j_{n},\ldots,j_{i+1}$ such that $\Nz(T_{i})\le p$.
We want to choose $j_{i}$ such that 
$\Nz(T_{i-1})\le p$, where by definition
\[ T_{i-1}= \left(\prod_{\ell = i}^{n} B_{\ell}(\ell,j_{\ell})\right)\Det(B_{i-1}) = 
\left(\prod_{\ell = i}^{n} B_{\ell}(\ell,j_{\ell})\right)\Det( B_{i}[i \mid j_{i}]).\]
This can be done as follows. From the  cofactor expansion of $\Det(B_i)$,
we have
\begin{align*}
T_{i}
 = {\sum_{j=1}^{i} (-1)^{i+j}  \left(\prod_{\ell = i+1}^{n} B_{\ell}(\ell,j_{\ell})\right) B_{i}(i,j) \Det(B_{i}[ i \mid j]).}
\end{align*}
Since $\Nz(T_{i})\le p$,  at least one of 
the terms  in the sum must have its least significant $1$-bit at position at most $p$. Thus, 
we can choose $j_i$  such that 
\[\Nz\left( \left(\prod_{\ell = i+1}^{n} B_{\ell}(\ell,j_{\ell})\right)
B_{i}(i,j_{i}) \Det(B_{i}[ i \mid j_{i}])\right) = \Nz(T_{i-1})\]
is minimized, which guarantees that $\Nz(T_{i-1}) \le p$.

Since by assumption $\Nz(T_n) = \Nz(\Det(B_n)) = p$,  $\VPV$
proves by $\Sigma_{0}^{B}(\LFP)$ induction on $i=n,\ldots,1$ that 
\[\Nz(T_{i})\le p.\]
If we define $j_1 =1$, then when $i=1$ we have
\[
T_{1} =\left(\prod_{\ell = 2}^{n} B_{\ell}(\ell,j_{\ell})\right)\Det(B_{1}) = \prod_{\ell = 1}^{n} B_{\ell}(\ell,j_{\ell}).\]
Thus it follows that $\Nz(T_{1}) = \Nz\left(\prod_{\ell = 1}^{n} B_{\ell}(\ell,j_{\ell})\right)\le p$.

Similarly to the proof of \theoref{Edmonds}, we can extract a perfect matching with weight at most $p$ by 
letting  $Q$ be a matrix, where $Q(i,j) = j$ for all $i,j\in [n]$. 
Then we compute another sequence of matrices 
\[Q_{n},Q_{n-1},\ldots, Q_{1},\]
where $Q_{n}=Q$ and $Q_{i-1}=Q_{i}[i\mid j_{i}]$, i.e., we delete from $Q_{i}$ exactly the 
row and column we deleted from $B_{i}$. 

To prove that $M=\set{(\ell,Q_{\ell}(\ell,j_{\ell}))\mid 1\le \ell\le n}$ is a perfect matching, we note 
that whenever a pair $(i,k)$ is added to the matching $M$, we  delete the row $i$ and column $j_{i}
$, where $j_{i}$ is the index satisfying $Q_{i}(i,j_{i}) = k$. So we can never match any other vertex to $k$ again.

It remains to show that $w(M)\le p$. 
Since \[\prod_{\ell = 1}^{n} B_{\ell}(\ell,j_{\ell})=2^{w(M)},\] the binary expansion of 
$\prod_{\ell = 1}^{n} B_{\ell}(\ell,j_{\ell})$ has a unique one  at position $w(M)$ and zeros 
elsewhere. Thus it follows from how the matching $M$ was  constructed that $w(M)\le p$. 
\end{proof}

The next two theorems complete our description and justification of our
$\RDET$ algorithm for finding a perfect matching.  For these theorems
we are given  a bipartite  graph $G=(U\uplus V,E)$, where we have $U=\set{u_1,\ldots,u_n}$ and  $V=\set{v_1,\ldots,v_n}$,
and each edge $(i,j)\in E$ is assigned a weight $w_{i,j}$ such that $G$ has  a \emph{unique}
minimum-weight perfect matching (see Theorem \ref{theo:iso}).
 Let $\vec{W}_{n\times n}$ be a sequence
satisfying $W_{i,j} = 2^{w_{i,j}}$ for all $(i,j)\in E$. Let $A$ be  
the Edmonds matrix of $G$, and let $B=A(\vec{W})$. Let $M$ denote the unique  minimum 
weight perfect  matching of $G$. 

\begin{thm} ($\VPV \proves$) \label{theo:Bin2}
The weight $p=w(M)$ is exactly $\Nz(\Det(B))$.
\end{thm}

If in \lemref{Bin1} we tried to extract an appropriate nonzero diagonal of $B$  using the determinant
and minors of  $B$ as our guide, then in the proof 
of this theorem we do the reverse. From a minimum-weight perfect matching $M$ of $G$, we want 
to rebuild in polynomially many steps suitable minors of $B$ until we fully recover the 
determinant of  $B$.  
We can then prove by  $\Sigma_{0}^{B}(\LFP)$ induction that in every step 
of this process,  each  ``partial determinant'' of  $B$ has  the least significant 1-bit at position $p$. 
Note that the technique we used to prove this theorem does have some similarity to that of
\lemref{permatrix}, even though the proof of this theorem is more complicated.

\begin{proof} 
Let $Q$ be a matrix, where $Q(i,j) = j$ for all $i,j\in [n]$.
For $1\le i\le n$ let $B_i$ be the result of deleting rows
$i+1,\ldots,n$ and columns $M(i+1),\ldots,M(n)$ from $B$ and let $Q_i$
be $Q$ with the same rows and columns deleted.  We can construct
these matrices inductively in the form of two matrix sequences
\begin{align*}
&B_{n},B_{n-1},\ldots, B_{1}		&&Q_{n},Q_{n-1},\ldots, Q_{1}
\end{align*}
where 
\begin{itemize}
\item we let $B_{n}=B$  and $Q_{n}=Q$, and 
\item for $i=n,n-1,\ldots,2$, define $j_{i}$ to be the unique index satisfying
\[ M(i) = Q_{i}(i,j_{i}),\] 
and then let
$B_{i-1} = B_{i}[i\mid j_{i}]$  and $Q_{i-1} =Q_{i}[i\mid j_{i}]$.
\end{itemize}
Then (setting $j_1=1$)
\begin{equation}\label{e:Bident}
   B_i(i,j_i) = B(i,M(i)) = 2^{w_{i,M(i)}}, \ 1\le i \le n
\end{equation}

\begin{myclaimb}
 $\Nz(\Det(B_i)) = \sum_{\ell = 1}^i w_{\ell,M(\ell)}$ for all $i\in [n]$.
\end{myclaimb}

The theorem follows from this by setting $i=n$. We will prove the claim by induction on $i$.
The base case $i=1$ follows from (\ref{e:Bident}).

For the induction step, it suffices to show 
$$
    \Nz(\Det(B_{i+1})) = \Nz(\Det(B_i)) + w_{i+1,M(i+1)}
$$

From the cofactor expansion formula we have
\begin{align*}
\Det(B_{i+1}) 
= \sum_{j=1}^{i+1}  (-1)^{(i+1)+j} B_{i+1}(i+1,j)\Det(B_{i+1}[ i+1 \mid j])
\end{align*}
Since  $B_{i+1}(i+1,j_{i+1}) = 2^{w_{i+1,M(i+1)}}$
by (\ref{e:Bident}),
and $\Det(B_{i+1}[ i+1 \mid j_{i+1}]) = \Det(B_i)$,
it suffices to show that if
\[R:=B_{i+1}(i+1,j_{i+1})\Det(B_{i+1}[ i+1 \mid j_{i+1}])\] 
then
\[
 \Nz(R) < \Nz\bigl(B_{i+1}(i+1,j)\Det(B_{i+1}[ i+1 \mid j])\bigr)
\]
for all $j\not=j_{i+1}$. 

Suppose for a contradiction that there is some  $j'\not=j_{\ell}$ such that 
\[
\Nz\bigl(B_{i+1}(i+1,j')\Det(B_{i+1}[ i+1 \mid j'])\bigr) \le \Nz(R).
\]   
Then, we can extend  the  set of edges
\[\bset{(n,M(n)),\ldots, (i+2,M(i+2)),(i+1,j')}\] 
with $i$  edges extracted from $B_{i+1}[ i+1 \mid j']$ (using the method from \lemref{Bin1})  to get a 
perfect matching of $G$ with weight at most $p$, which contradicts that $M$ is the unique 
minimum-weight perfect matching of $G$.
\end{proof}

To extract the edges of  $M$ in $\DET$, we need to decide if 
an edge $(i,j)$ belongs to  the unique minimum-weight perfect matching $M$ without knowledge
of other edges in $M$.  The next theorem, whose proof follows directly from \lemref{Bin1} and \theoref{Bin2}, gives us that method. 

\begin{thm}($\VPV \proves$)\label{theo:Bin3}
For every edge  $(i,j)\in E$, we have $(i,j)\in M$ if and only if 
\[w (M)-w_{i,j}=\Nz\bigl(\Det (B[i\mid j])\bigr).\]
\end{thm}
\begin{proof} ($\Rightarrow$): Assume  $(i,j)\in M$. Then the bipartite graph 
$G' = G\setminus\set{u_{i},v_{j}}$  must have a unique minimum-weight perfect matching 
of weight  $w(M)-w_{i,j}$. Thus from  \theoref{Bin2}, 
\[\Nz\bigl(\Det (B[i\mid j])\bigr)=w(M)-w_{i,j}.\]

 ($\Leftarrow$):  We prove the contrapositive.
 Assume $(i,j)\not\in M$. Suppose for a contradiction that 
\[w (M)-w_{i,j}=\Nz\bigl(\Det (B[i\mid j])\bigr).\] 
 Then by \lemref{Bin1} we can 
extract from the submatrix $B[i\mid j]$ a perfect matching $Q$ of the bipartite graph 
$G' = G\setminus\set{u_{i},v_{j}}$ with weight at most $w(M)-w_{i,j}$. But then  
$M' = Q\cup \set{(i,j)}$ is another perfect matching  of $G$ with $w(M')\le w(M)$, a contradiction.
\end{proof}

Theorems \ref{theo:iso}, \ref{theo:Bin2}, and \ref{theo:Bin3}
complete the description
and justification of our $\RDET$ algorithm for finding a perfect
matching in a bipartite graph.  Since these are theorems of $\VPV$,
it follows that $\VPV$ proves the correctness of the algorithm.

\subsection{Related bipartite matching problems}
The correctness of the Mulmuley-Vazirani-Vazirani algorithm can easily be used to establish the 
correctness of $\RDET$ algorithms for related matching problems, for example,  the maximum (cardinality) bipartite matching problem and the minimum-weight  bipartite perfect matching problem, where the weights
assigned to the edges are small.  We refer to \cite{MVV87} for more details on these reductions.

\section{Conclusion and Future Work}
We have only considered randomized matching algorithms for {\em bipartite}
graphs. For general 
undirected graphs, we need Tutte's matrix (cf. \cite{MVV87}), a generalization of 
Edmonds' matrix. Since  every Tutte matrix is  a skew symmetric matrix
where each variable appears
exactly twice, we cannot directly  apply our technique for Edmonds'
matrices, where each variable appears at most once. However,
by using the recursive definition of the Pfaffian 
instead of the cofactor expansion, we believe that it is also possible to
generalize our results to general undirected graphs. 
We also note that the Hungarian algorithm only works for weighted
bipartite graphs.   To find a maximum-weight matching of a weighted
undirected graph, we need to formalize Edmonds'
\emph{blossom algorithm} (cf. \cite{KV08}). Once we have the correctness of the  blossom algorithm, 
the proof of the Isolating Lemma for undirected graph perfect matchings
will be the same as that of \theoref{iso}. We leave the detailed
proofs for the general undirected graph case for  future work.

It is worth noticing that symbolic determinants of Edmonds' matrices
result in very special polynomials, whose
structures can be used to define the $\VPV$ surjections witnessing the
probability bound in the Schwartz-Zippel Lemma as demonstrated in this
paper. It remains an open problem whether we can prove the full version
of the Schwartz-Zippel Lemma using Je\v r\'abek's method within the theory $\VPV$.

We have shown that the correctness proofs for several randomized
algorithms can be formalized in the theory $\VPV$ for polynomial time
reasoning.  But some of the algorithms involved are in the subclass
$\DET$ of polynomial time, where $\DET$ is the closure of $\#L$
(and also of the integer determinant function $\Det$) under $\AC^0$
reductions.  As mentioned in Section~\ref{s:basicBA} the 
subtheory $\overline{\VsL}$ of $\VPV$ can define all functions in
$\DET$, but it is open whether $\overline{\VsL}$ can prove
properties of $\Det$ such as the expansion by minors.  However the
theory $\overline{\VsL} + \CH$ can prove such properties, where
$\CH$ is an axiom stating the Cayley-Hamilton Theorem.  
Thus in the statements of Lemma \ref{lem:permatrix} and
Theorem \ref{theo:Knn} we could have replaced $(\VPV\proves)$ by
$(\VsL + \CH\proves)$.  We could have done the same for Theorem \ref{theo:SZEd}
if we changed the argument $\vec{W}$ of the function $H$ to $M$,
where $M$ is a permutation matrix encoding a perfect matching for the
underlying bipartite graph. 
This modified statement of the theorem still proves the interesting
direction of the correctness of the bipartite perfect matching
algorithm in Section \ref{s:formBi},
since the function $H$ is used only to bound the error assuming that $G$
does have a perfect matching.

We leave open the question of whether any of the other correctness
proofs can be formalized in $\overline{\VsL} + \CH$.

We believe that Je\v r\'abek's framework deserves to be studied in greater depth since it  helps us to understand better  the connection between probabilistic reasoning and weak systems of bounded arithmetic.
We are working on  using  Je\v r\'abek's ideas to formalize 
constructive aspects of fundamental theorems in finite probability in
the spirit of the recent beautiful work by Moser and Tardos \cite{MT10}, Impagliazzo and Kabanets \cite{IK10}, etc.

\appendix
\section{Formalizing the Hungarian algorithm}
Before proceeding with the  Hungarian algorithm, we need to formalize the two most fundamental 
theorems for bipartite matching: Berge's Theorem and Hall's Theorem.

\subsection{Formalizing Berge's Theorem and the augmenting-path algorithm\label{app:Berge}} 
Let $G=(X \uplus Y, E)$ be  a bipartite graph, where $X=\set{x_{i}\mid 1\le i\le n}$ and $Y=\set{y_{i}
\mid 1\le i\le n}$. Formally to make sure that $X$ and $Y$ are disjoint, we can let $x_{i}:=i$ and $y_
{i}:=n+i$. We encode the edge relation $E$  of $G$ by a matrix $E_{n\times n}$, where $E(i,j)
=1$ iff $x_{i}$ is adjacent to $y_{j}$. Note that we often abuse notation and write $\set{u,v}\in E$ to 
denote that $u$ and $v$ are adjacent in $G$, which formally means either 
\begin{align*}
&u\in X\;\wedge\; v\in Y \wedge E(u, v-n), \text{ or} \\
&v\in X\;\wedge\; u\in Y \wedge E(v, u-n). 
\end{align*}
This complication is due to the usual convention of using an $n\times n$ matrix to encode the edge 
relation of a bipartite graph with $2n$ vertices.

An $n\times n$ matrix $M$ encodes a matching of $G$ iff  $M$ is a permutation matrix satisfying 
\[\forall i,j\in [n], M(i,j)\rightarrow E(i,j).\]

We represent a path by a sequence of vertices $\seq{v_1,\ldots,v_k}$ with  $\set{v_i,v_{i+1}}\in E$ 
for all $i \in [k]$.  

Given a matching $M$, a vertex $v$ is $M$-\emph{saturated} if $v$ is incident with an edge in $M$. 
We will say $v$ is $M$-\emph{unsaturated} if it is not $M$-saturated. A path $P=\seq
{v_1,\ldots,v_k}$ is an $M$-\emph{alternating path} if  $P$ alternates between edges in $M$ and 
edges in $E\setminus M$. More formally, $P$ is an $M$-alternating path if either  of the following 
two conditions holds:
\begin{itemize}
 \item For every $i\in \set{1,\ldots, k-1}$, $\set{v_i,v_{i+1}}\in E\setminus M$ if $i$ is odd, and $\set
{v_i,v_{i+1}}\in M$ if $i$ is even.
 \item  For every $i\in \set{1,\ldots, k-1}$,  $\set{v_i,v_{i+1}}\in E\setminus M$ if $i$ is even, and $\set
{v_i,v_{i+1}}\in M$ if $i$ is odd.
\end{itemize}

An $M$-alternating path $\seq{v_1,\ldots,v_k}$ is an $M$-\emph{augmenting path}  if   the vertices $v_{1}$ and $v_{k}$ are $M$-unsaturated.

\begin{thm}[Berge's Theorem]\label{theo:aug}
($\VPV\proves $) Let $G=(X \uplus Y, E)$ be  a bipartite graph. A 
matching $M$ is maximum iff there is no $M$-augmenting path in $G$.
\end{thm}
\begin{proof}($\Rightarrow$): Assume that all matchings $N$ of $E$ satisfy $|N|\le |M|$. Suppose for 
a contradiction that there is an $M$-augmenting path $P$. Let $M \oplus P$ denote the symmetric difference 
of two sets of edges $M$ and $P$. Then $M'= M \oplus P$ is a matching  
greater than $M$, a contradiction. 

($\Leftarrow$): We will prove the contrapositive. Assume there is another matching $M'$ satisfying 
$|M'|> |M|$. We want to construct an $M$-augmenting path in $G$.

Consider $Q = M' \oplus M$. Since $|M'|> |M|$, it follows that $|M' \setminus M| > |M\setminus M'|$, and thus 
\begin{align}
|Q\cap M'|>|Q\cap M| 	\label{eq:aug.1}
\end{align}
Note that we can compute cardinalities of the sets directly here since all the sets we are 
considering here are small. Now let $H$ be the graph whose edge relation is $Q$ and whose 
vertices are simply  the vertices of $G$. We then observe the following properties of $H$:
\begin{itemize}
 \item Since $Q$ is constructed from two matchings $M$ and $M'$, every vertex of $H$ can only be 
incident with at most two edges: one from $M$ and another from $M'$. So every vertex of $H$ has 
degree at most $2$.
 \item Any path of $H$ must alternate between the edges of $M$ and $M'$.
\end{itemize}
We will provide a polytime algorithm to extract  from the graph $H$ an augmenting path with 
respect to $M$, which gives us the contradiction.
\begin{algorithmic}[1]
\STATE Initialize $K= H$ and  $i= 1$
\WHILE{$K\not= \emptyset$}
	\STATE Pick the least vertex $v\in K$
	\STATE Compute the connected component $C_{i}$ containing $v$
	\IF{$C_{i}$ is an $M$-augmenting path}
	\STATE return $C_{i}$ and halt.
	\ENDIF
	\STATE Update $K = K \setminus C_{i}$ and $i= i+1$.
\ENDWHILE
\end{algorithmic}

Note that since $H$ has $n$ vertices, the while loop can only iterate at most $n$ times. It only 
remains to show the following.

\begin{myclaimb}
The algorithm returns an $M$-augmenting path assuming
$|M'| > |M|$.
\end{myclaimb}

Suppose for a contradiction that the algorithm would never produce any $M$-augmenting path. 
Since $H$ has degree at most two, in every iteration of the while loop, we know that 
the connected component $C_{i}$ is 
\begin{itemize}
\item either a cycle, which means $|C_{i}\cap M| = |C_{i}\cap M'|$, or 
\item a path but not an $M$-augmenting path, which implies that  $|C_{i}\cap M| \ge  |C_{i}\cap M'|$.
\end{itemize}
Since $Q=M' \oplus M = \bigcup_{i} C_{i}$ and all $C_{i}$ are disjoint, we have 
\[|Q\cap M| = |\bigcup_{i} (C_{i} \cap M)| \ge |\bigcup_{i} (C_{i} \cap M')| = |Q\cap M'|.\]
But this contradicts \eref{aug.1}.
\end{proof}

\begin{algo}[The augmenting-path algorithm] As a corollary of Berge's Theorem, we have the 
following simple algorithm for finding a maximum matching of a bipartite graph $G$. We start from 
any matching $M$ of $G$, say empty matching. Repeatedly locate an $M$-augmenting path $P$ 
and augment $M$ along $P$ and replace $M$ by the resulting matching. Stop when there is no $M
$-augmenting path. Then we know that $M$ is maximum. 
Thus, it remains to show how to search for an $M$-augmenting path given a matching $M$ of $G$.
\end{algo}

\begin{algo}[The augmenting-path search algorithm]\label{alg:MWP} 
First, from $G$ and $M$ we construct a directed graph $H$, where the vertices $V_{H}$ of $H$ are 
exactly the vertices $X\uplus Y$ of $G$, and the edge relation $E_{H}$ of $H$ is a $2n\times 2n$ 
matrix  defined as follows:
\begin{align*}
E_{H} := \bset{(x,y)\in X\times Y \mid \set{x,y}\in E\setminus M} 
		\cup \bset{(y,x)\in Y\times X\mid \set{y,x}\in M}.
\end{align*} 
The key observation is that $t$ is reachable from $s$ by an $M$-alternating path in the bipartite 
graph $G$ iff $t$ is reachable from $s$ in the directed graph $H$.

After constructing the graph $H$, we can  search for an $M$-augmenting path using the \emph
{breadth first search} algorithm as follows. Let $s$ be an $M$-unsaturated  vertex in $X$. We 
construct two $2n\times 2n$ matrices $S$ and $T$ as follows.
\begin{algorithmic}[1]
\STATE The row $\row(1,S)$ of $S$ encodes the set $\set{s}$, the starting point of our search.
\STATE  From $\row(i,S)$, the set $\row(i+1,S)$  is defined as follows:
$j\in \row(i+1,S) \leftrightarrow$ there exists some  $k\in [n]$, \[\row(i,S)(k) \;\wedge\; E_{H}(k,j) \;
\wedge \forall \ell\in [i], \neg \row(\ell,S)(k).\]
After finish constructing $\row(i+1,S)$, we can update $T$ by setting $T(j,k)$ to $1$ for every $k\in 
\row(i,S)$ and $j\in \row(i+1,S)$ satisfying $E_{H}(k,j)$. \\[0.5cm]
\end{algorithmic}

\noindent Intuitively, $\row(i,S)$ encodes the set of vertices that are of distance $i-1$ from $s$, and $T$ is the 
auxiliary matrix that can be used to recover a path from $s$ to any vertex $j\in \row(i,S)$ for all $i\in 
[2n]$.

\end{algo}
\begin{rem} Let $R=\bigcup_{i=2}^{n}\row(i,S)$, then $R$ the set of vertices reachable from $s$ 
by an $M$-alternating path. This follows from the fact that our construction mimics the construction 
of the formula $\delta_{\textsf{CONN}}$, which was used to define the theory $\VNL$ for $\NL$ in 
\cite[Chapter 9]{CN10}. The matrix $T$ is constructed to help us trace a path for every vertex $v\in 
R$ to $s$.
\end{rem}

By induction on $i$, we can prove that $\row(i,S)\subseteq X$ for every odd $i\in [2n]$, and $\row(i,S)
\subseteq Y$ for every even $i\in [2n]$. Thus, after having $S$ and $T$, we choose the largest even 
$i^{*}$ such that $\row(i^{*},S)\not= \emptyset$, and then search the set $\row(i^{*},S)\cap Y$ for an 
$M$-unsaturated vertex $t$. If such vertex $t$ exists, then we use $T$ to trace back a path to $s$. 
This path will be our $M$-augmenting path. If no such vertex $t$ exists, we report that there is no 
$M$-augmenting path.

\subsection{Formalizing Hall's Theorem \label{app:hall}}
\begin{thm}[Hall's Theorem]($\VPV\proves $)  Let $G=(X \uplus Y, E)$ be  a bipartite graph. 
Then $G$ has a perfect matching iff for every subset $S\subseteq X$, $|S|\le |N(S)|$, where $N(S)$ 
denotes the neighborhood of $S$ in $G$.
\end{thm}

The condition $\forall S\subseteq X, |S|\le |N(S)|$, which is necessary and sufficient for a bipartite 
graph to have a perfect matching, is called Hall's condition. We encode a set $S\subseteq X$ in the 
theorem as a binary string of length $n$, where $S(i) = 1$ iff $x_{i}\in S$. Similarly, we encode the 
neighborhood $N(S)$ as a binary string of length $n$, and we define 
\[N(S):= \bigcup\bset{\row(i,E)\mid x_{i}\in S},\] 
where the union can be computed by taking the disjunction of all binary vectors in the set
\[\bset{\row(i,E)\mid x_{i}\in S}\] 
componentwise. Note that we can compute the cardinalities of $S$ and $N(S)$ 
directly since both of these sets are subsets of small sets $X$ and $Y$.

\begin{proof}($\Rightarrow$): Assume $M$ is a perfect matching of $G$. 
Given a subset $S\subseteq X$, the vertices in $S$ are matched to 
the vertices in some subset $T\subseteq Y$ by the perfect matching $M$, where $|S|=|T|$. Since 
$T\subseteq N(S)$, we have $|N(S)|\ge |T|  = |S|$. 

($\Leftarrow$): We will prove the contrapositive. Assume $G$ does not have a perfect matching, we 
want to construct a subset $S\subseteq X$ such that $|N(S)|< |S|$. 

Let $M$ be a maximum but not perfect matching constructed by the augmenting-path algorithm. Since $M$ is not a perfect 
matching, there is some $M$-unsaturated vertex $s\in X$. Let $S$ and $T$ be the result of running 
the  ``augmenting-path search'' algorithm from $s$, then $R := \bigcup_{i=2}^{n}\row(i,S)$ is the 
set of all vertices reachable from $s$ by an $M$-alternating path. Since there is no $M$-augmenting path, all the vertices in  $R$ are $M$-saturated. We want to show the following two  
claims.

\begin{myclaim}{1}
The vertices in $R\cap X$ are all matched to the vertices in $R\cap Y$ by $M$, 
and 
\[|R\cap X| = |R\cap Y|.\]
\end{myclaim}

Suppose for a contradiction that some vertex $v\in R$  is not matched to any vertex $u\in R$ by 
$M$. Since we already know that all vertices in $R$ are $M$-saturated, $v$ is matched by some 
vertex $w\not\in R$ by $M$. But this is a contradiction since $w$ must be reachable from $s$ by an 
alternating path, and so  the augmenting-path search algorithm must already have added $w$ to 
$R$. Thus, the vertices in $R\cap X$ are all matched to the vertices in $R\cap Y$ by $M$, which 
implies that $|R\cap X| = |R\cap Y|$.

\begin{myclaim}{2}
$N(R\cap X) = R\cap Y$. 
\end{myclaim}

Since $R\cap X$ are matched to $R\cap Y$, we know $N(R\cap X) \supseteq R\cap Y$. Suppose 
for a contradiction that $N(R\cap X) \supset R\cap Y$. Let $v\in N(R\cap X) \setminus R\cap Y$, 
and $u\in R\cap X$ be the vertex adjacent to $v$. Since $u$ is reachable from $s$ by an 
$M$-alternating path $P$, we can extend $P$ to get an $M$-alternating path from $s$ to $v$, which 
contradicts that $v$ is not added to $R$. 

We note that $N(\set{s}\cup (R\cap X)) = R\cap Y$. Then $S = \set{s}\cup (R\cap X)$  is the 
desired set since
\[|N(S)| = |R\cap Y| =  |R\cap X| < |S|.\]
\end{proof}

From the proof of Hall's Theorem, we have the following corollary saying that if a bipartite graph 
does not have a perfect matching, then we can find in polytime a subset of vertices violating Hall's 
condition.

\begin{corollary}($\VPV\proves $) \label{cor:hall}
There is a $\VPV$ function that, on input a bipartite graph $G$ that does not have a perfect matching, 
outputs  a subset $S\subseteq X$ such that $|S|> |N(S)|$.
\end{corollary}

\subsection{Proof of \theoref{hung} \label{app:hung}}
Let $H =H_{\vec{u},\vec{v}}$ be the equality subgraph for the weight cover  $(\vec{u},\vec{v})$, and 
let $M$ be a maximum cardinality matching of $H$. Recall \theoref{hung} wants us  to show that  $\VPV$ 
proves equivalence of the following three statements:
\begin{enumerate}[(1)]
\item $w(M) = \cost(\vec{u},\vec{v}) $
\item $M$ is a maximum-weight matching and the cover $(\vec{u},\vec{v})$ is a minimum-weight 
cover of $G$
\item $M$ is a perfect matching of $H$
\end{enumerate}

\begin{proof}[Proof of \theoref{hung}](1)$\Rightarrow$(2):
Assume that $\cost(\vec{u},\vec{v}) = w(M)$. By \lemref{hung}, no matching has weight greater than 
$\cost(\vec{u},\vec{v})$, and no cover with weight less than $w(M)$. 

(2)$\Rightarrow$(3):  Assume $M$ is a maximum-weight matching and $(\vec{u},\vec{v})$ is a 
minimum-weight cover of $G$. Suppose for a contradiction that the maximum matching $M$ is not 
a perfect matching of $H$. We will construct a weight cover whose cost is strictly less than 
$\cost(\vec{u},\vec{v})$, which contradicts that  $(\vec{u},\vec{v})$ is a minimum-weight cover.

Since the maximum matching $M$ is not a perfect matching of $H$, by \corref{hall}, we can 
construct in polytime a subset  $S\subseteq X$ satisfying \[|N(S)|<|S|.\] 
Then we calculate the quantity
\[\delta = \min\set{u_{i}+v_{j} - w_{i,j}\mid x_{i}\in S \;\wedge\; y_{j}\not\in N(S) }.\]
Note that $\delta>0$ since $H$ is the equality subgraph.
Next we construct  a  pair of sequences  $\vec{u}'=\seq{u'_{i}}_{i=1}^{n}$ and $\vec{v}'=\seq{v'_{i}}_{i=1}
^{n}$, as follows:
\begin{align*}
 u'_{i} = \begin{cases}
u_{i} -\delta 	& \text{ if } x_{i}\in S\\
u_{i}	 		& \text{ if } x_{i}\not\in S
\end{cases} &&
 v'_{i} = \begin{cases}
v_{j} +\delta 	& \text{ if } y_{j}\in N(S)\\
v_{j}	 		& \text{ if } y_{j}\not\in N(S)
\end{cases}
\end{align*}
We claim that $(\vec{u}',\vec{v}')$ is  again a weight cover. The condition $w_{i,j}\le  u'_{i} + v'_{j}$ 
might only be violated for $x_{i}\in S$ and $y_{i}\not\in N(S)$. But since we chose $\delta \le u_{i}
+v_{j} - w_{i,j}$, it follows that  \[w_{i,j}\le (u_{i}-\delta)+v_{j} = u'_{i} + v'_{j}.\] 
Since
\begin{align*}
\cost(\vec{u}',\vec{v}')&={\textstyle  \sum_{i=1}^{n} (u'_{i} + v'_{i})} \\
&=   \underbrace{{\textstyle \sum_{i=1}^{n}} (u_{i} + v_{i})}_{=\cost(\vec{u},\vec{v})} + \underbrace
{\delta|N(S)| - \delta |S|}_{<0},
\end{align*}
it follows that $\cost(\vec{u}',\vec{v}')<\cost(\vec{u},\vec{v})$.

(3)$\Rightarrow$(1): Suppose $M$ is a perfect matching of $H$.  Then $w_{i,j}= u_{i} + v_{j}$ holds 
for all edges in $M$. Summing equalities $w_{i,j}= u_{i} + v_{j}$ over all edges of $M$ yields the 
equality $\cost(\vec{u},\vec{v}) = w(M)$.
\end{proof}

\section*{Acknowledgement} We would like to thank Emil Je\v{r}\'abek for answering our questions regarding his framework. We also thank the referees and Lila Fontes for their constructive comments. This work was financially  supported by the Ontario Graduate Scholarship and the Natural Sciences and Engineering Research Council of Canada.

\bibliographystyle{plain} 
\bibliography{wphp}


\end{document}